%% file: templatePRIME.tex
\documentclass{article}

\usepackage{PRIMEarxiv}

\usepackage[utf8]{inputenc} 
\usepackage[T1]{fontenc}    
\usepackage{hyperref}       
\usepackage{url}            
\usepackage{booktabs}       
\usepackage{amsfonts}       
\usepackage{nicefrac}       
\usepackage{microtype}      
\usepackage{lipsum}
\usepackage{graphicx}
\graphicspath{{media/}}     

\input{extras}
\begin{document}

\title{Joint Supervised and Self-supervised Learning for MRI Reconstruction}

\author{George Yiasemis\\
Netherlands Cancer Institute\\
Amsterdam, Netherlands\\
\and
Nikita Moriakov\\
Netherlands Cancer Institute\\
Amsterdam, Netherlands\\
\and
Clara I. Sánchez\\
University of Amsterdam\\
Amsterdam, Netherlands\\
\and
Jan-Jakob Sonke\\
Netherlands Cancer Institute\\
Amsterdam, Netherlands\\
\and
Jonas Teuwen
\thanks{\tt\small Corresponding author: Jonas Teuwen (Email: j.teuwen@nki.nl)}
\\
Netherlands Cancer Institute\\
Amsterdam, Netherlands\\
}

\maketitle

\begin{abstract}
    Magnetic Resonance Imaging (MRI) represents an important diagnostic modality; however, its inherently slow acquisition process poses challenges in obtaining fully-sampled $k$-space data under motion. In the absence of fully-sampled acquisitions, serving as ground truths, training deep learning algorithms in a supervised manner to predict the underlying ground truth image becomes challenging. To address this limitation, self-supervised methods have emerged as a viable alternative, leveraging available subsampled $k$-space data to train deep neural networks for MRI reconstruction. Nevertheless, these approaches often fall short when compared to supervised methods. We propose Joint Supervised and Self-supervised Learning (JSSL), a novel training approach for deep learning-based MRI reconstruction algorithms aimed at enhancing reconstruction quality in cases where target datasets containing fully-sampled $k$-space measurements are unavailable. JSSL operates by simultaneously training a model in a self-supervised learning setting, using subsampled data from the target dataset(s), and in a supervised learning manner, utilizing datasets with fully-sampled $k$-space data, referred to as proxy datasets. We demonstrate JSSL's efficacy using subsampled prostate or cardiac MRI data as the target datasets, with  fully-sampled brain and knee, or brain, knee and prostate $k$-space acquisitions, respectively,  as proxy datasets. Our results showcase substantial improvements over conventional self-supervised methods, validated using common image quality metrics. Furthermore, we provide theoretical motivations for JSSL and establish ``rule-of-thumb" guidelines for training MRI reconstruction models. JSSL effectively enhances MRI reconstruction quality in scenarios where fully-sampled $k$-space data is not available, leveraging the strengths of supervised learning by incorporating proxy datasets.
\end{abstract}


\section{Introduction}
\label{sec:sec1}

Magnetic Resonance Imaging (MRI) is a widely used imaging modality in clinical practice due to its ability to non-invasively visualize detailed anatomical and physiological information inside the human body. However, the physics involved in the acquisition of MRI data, also known as the $k$-space, often makes it time-consuming, limiting its applicability in scenarios where fast imaging is essential, such as image-guided tasks. The MRI acquisition can be accelerated by acquiring subsampled $k$-space data, although this approach yields lower-quality reconstructed images with possible artifacts and aliasing \cite{zbontar2019fastmri}. 

In the past half decade, numerous state-of-the-art MRI reconstruction techniques have emerged that employ Deep Learning (DL)-based reconstruction algorithms \cite{fessler2019optimization, Pal2022-ce}. These algorithms are trained to produce high-quality images from subsampled $k$-space measurements, surpassing conventional reconstruction methods such as Parallel Imaging \cite{Pruessmann1999, Griswold2002} or Compressed Sensing \cite{1580791, 4472246}. Typically, these algorithms are trained in a fully supervised manner using retrospectively subsampled $k$-space measurements (or images) as inputs and fully-sampled $k$-space data (or images) as ground truth.

Despite the high performance of these methods, there are certain cases in clinical settings where acquiring fully-sampled datasets, essential for supervised learning (SL) training,  can be infeasible or prohibitively expensive \cite{Uecker2010, Sarma2014, Kim2021}. Such cases include MR imaging of the abdomen, cardiac cine, chest, or the prostate, where periodic or aperiodic motion can make it impossible to collect measurements adhering to the Nyquist-Shannon sampling theorem \cite{1697831}. 

In recent years, to overcome this challenge, several approaches have been proposed that train DL-based algorithms under self-supervised learning (SSL) settings, using the available (subsampled) $k$-space measurements without the need for ground truth fully-sampled data \cite{Yaman2020,Hu2021,zhou2022dsformer,cui2022selfscore,millard2023}. These methods harness self-supervised mechanisms to train models to reconstruct subsampled MRI data.

In this work, we introduce Joint Supervised and Self-supervised Learning (JSSL), a novel method for training DL-based MRI reconstruction models when ground truth fully-sampled $k$-space data for a target organ domain is unavailable for supervised training. JSSL leverages accessible fully-sampled data from \textit{proxy} dataset(s) and subsampled data from the \textit{target} dataset(s) to jointly train a model in both supervised and self-supervised manners. Our contributions can be summarized as follows:

\begin{itemize}[leftmargin=*]
    \item At the time of writing, our proposed JSSL method represents the first approach to combine supervised and self-supervised learning-based training, in proxy and target organ domains, respectively within a single pipeline in the context of accelerated DL-based MRI reconstruction.
    \item We provide a theoretical motivation for JSSL.
    \item We experimentally demonstrate that JSSL outperforms self-supervised DL-based MRI reconstruction approaches, with a specific focus on subsampled prostate and cardiac datasets.
    \item We offer practical ``rule-of-thumb" guidelines for selecting appropriate training frameworks for accelerated MRI reconstruction models.
\end{itemize}

The rest of the paper is organized as follows: \Section{sec2} reviews related work on SSL-based MRI reconstruction, which provides the background for our approach. \Section{sec3} introduces the concepts of supervised and self-supervised training for MRI reconstruction, laying the groundwork for our proposed method, along with the theoretical basis and details of our experimental setup.  \Section{sec4} presents the quantitative and qualitative results, and \Section{sec5} discusses these findings.

\begin{figure}[!htb]
    \centering
    \includegraphics[width=1\columnwidth]{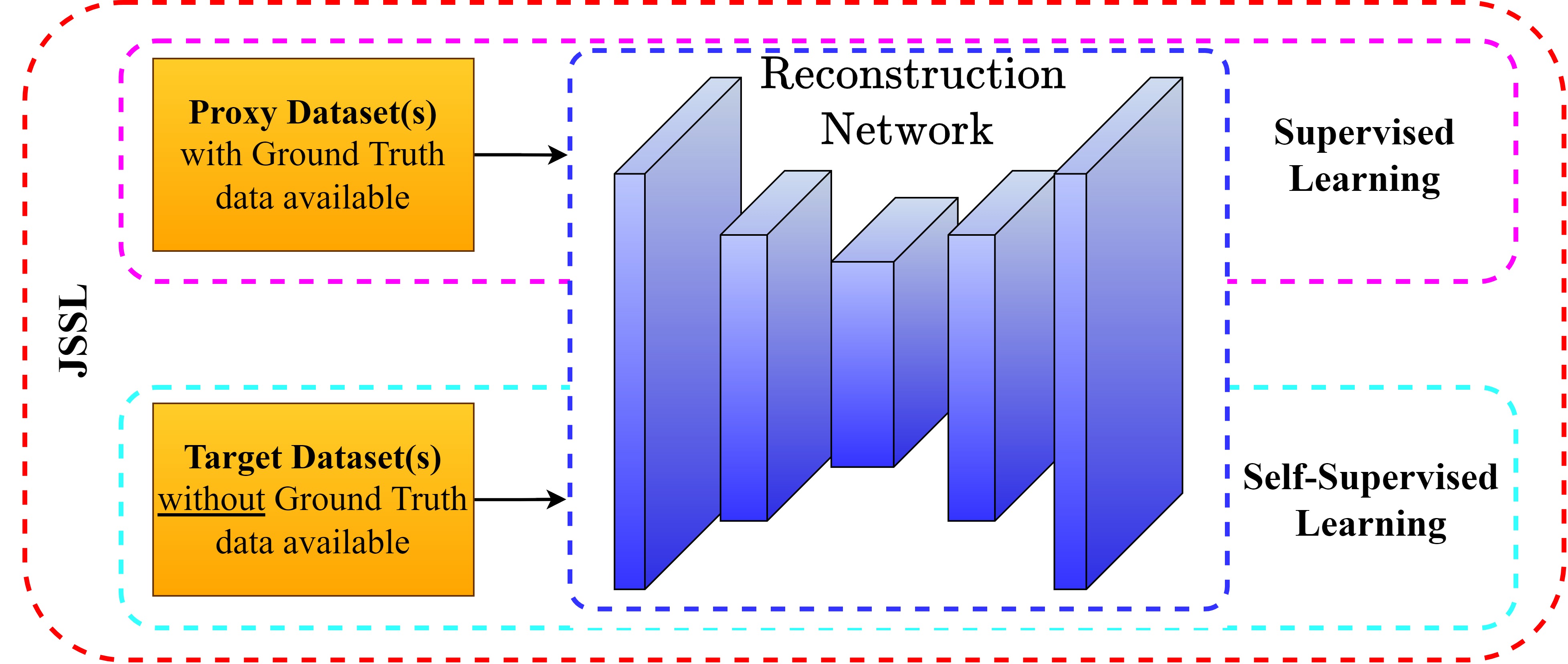}
    \caption{Overview of the JSSL framework for MRI reconstruction. Training uses ground truth data from proxy domain(s) and subsampled data (no ground truth) from a target domain, jointly in supervised and self-supervised manners, respectively.}
    \label{fig:overview}
\end{figure}

\section{Related Work}
\label{sec:sec2}
In the realm of self-supervised learning-based MRI reconstruction, among the first works introduced was SSDU (Self-supervised learning via data undersampling) \cite{Yaman2020}. SSDU, inspired by SSL concepts from deep learning, particularly Noise2Self \cite{batson2019noise2self}, proposed training a reconstruction model (ResNet CNN with conjugate gradient formulation) by partitioning the undersampled data into two subsets. One subset served as input, and the other as the target, with the loss estimated in the $k$-space domain.

An extension of SSDU was proposed in a parallel network framework \cite{Hu2021}, where two networks were trained on each partition of the subsampled $k$-space data. A consistency loss minimized the discrepancy between the two networks' outputs, allowing either network to be used during inference since both networks were trained to produce consistent results.

Further building on SSDU, \cite{millard2023} introduced a Noisier2Noise framework,  where a second subsampling mask was applied to the already subsampled $k$-space data. The employed network, E2EVarNet \cite{Sriram2020}, was trained to recover singly subsampled data from the doubly subsampled version, showing that SSDU is a special case of this broader method.  Furthermore, \cite{millard2023} provided theoretical justifications for SSDU.

In the realm of diffusion-based MRI reconstruction, a fully-sampled-data-free score-based diffusion model was proposed in \cite{cui2022selfscore}, where the model learned the prior of fully-sampled images from subsampled data in a self-supervised manner. Another diffusion-based approach, SSDiffRecon \cite{korkmaz2023self}, integrated cross-attention transformers with data-consistency blocks in an unrolled architecture. However, diffusion-based methods are outside the scope of our work.

Following the SSDU subsampled data splitting, in \cite{yan2023dc} the authors present DC-SiamNet, which employs two branches with shared weights in a Siamese architecture. Each branch reconstructs from a partition of the $k$-space data, and the training is guided by a dual-domain loss that includes image and frequency domain consistency which ensure reconstructed images/$k$-spaces are consistent across partitions, along with contrastive loss in the latent space. 

A more recent work extended SSDU by introducing SPICER, which includes coil sensitivity estimation based on autocalibration signal (ACS) data and utilizes U-Net-based models for both sensitivity estimation and reconstruction \cite{hu2024spicerselfsupervisedlearningmri}. Similar sensitivity estimation was also employed in \cite{millard2023} within the E2EVarNet framework.

Finally, SSDU has also been applied to reconstruct non-Cartesian MRI data, with the subsampled $k$-space split into disjoint parts \cite{ZHOU2022102538}. In this approach, a variational network is trained using a dual-domain loss similar to \cite{yan2023dc}: frequency consistency ensures that reconstructed $k$-spaces from each partition match the input data, while image consistency ensures that the reconstructed images are consistent across partitions. Additionally, loss is computed by comparing the reconstructed $k$-spaces and images from each partition with those generated when subsampled data is used as input.

Most self-supervised MRI reconstruction methods can be seen as derivatives or extensions of SSDU, with partitioning of undersampled data into disjoint subsets as the fundamental idea. This partitioning approach underpins the SSL component of our method, and without loss of generality, SSDU can be considered a representative method in this domain. While recent techniques have incorporated different architectures or loss functions, they largely build upon this core strategy.

Our proposed method, Joint Supervised and Self-Supervised Learning, draws inspiration from these aforementioned approaches. Like most SSL-based methods, it seeks to overcome the challenge of training without fully-sampled $k$-space data for the target organ domain. However, JSSL extends the applicability of such techniques by leveraging fully-sampled data from proxy datasets while incorporating subsampled data from the target domain. This enables joint training through both supervised and self-supervised learning, providing a practical solution for scenarios where ground truth fully-sampled data is inaccessible, yet allowing for improved reconstruction performance through the combination of proxy and target datasets.

In the broader context of combining supervised and self-supervised learning, Noise2Recon \cite{desai2021noise2recon} extended SSDU by leveraging both fully-sampled and subsampled data within a single organ domain for reconstruction and denoising, using the E2EVarNet model \cite{Sriram2020}. However, this method's dependency on fully-sampled data restricts its applicability in scenarios where such data is unavailable.

Another recent approach utilized paired fully-sampled and subsampled data from different modalities for reconstruction of the target modality \cite{zhou2022dsformer}. While SSL was employed for training, this method still relied on fully-sampled data during both training and inference, which contrasts with our approach that focuses on cases where fully-sampled data is unavailable for the target domain.

Lastly, test-time training \cite{darestani2022testtimetrainingclosenatural} is a recent method proposed to handle domain shifts in MRI reconstruction. By re-training models at inference times using a SSL data-consistency loss, it aims to adjust to shifts in data distribution between training and testing, such as moving from one scanner to another. However, this technique operates at inference time, which limits its utility in real-time imaging applications.

\section{Materials and Methods}
\label{sec:sec3}

\subsection{MRI Acquisition and Reconstruction}
\label{sec:sec3.1}

Assuming a fully-sampled MRI acquisition, the ground truth image $\vec{x} \in \mathbb{C}^{n}$ can be recovered from the fully-sampled multi-coil ($n_c>1$) $k$-space $\vec{y} \in \mathbb{C}^{n\times n_c}$ by applying the inverse Fast Fourier transform (FFT), denoted as $\mathcal{F}^{-1}$, followed by the root-sum-of-squares (RSS) method:

\begin{equation}
\begin{gathered}
    \vec{x} = \text{RSS} \circ \mathcal{F}^{-1} (\vec{y})=\left(\sum_{k=1}^{n_c}|\mathcal{F}^{-1}(\vec{y}^k)|^2\right)^{1/2},
    \label{eq:rss}
\end{gathered}
\end{equation}

\noindent
or alternatively, with known sensitivity maps $\vec{S}$, the SENSE operator:

\begin{equation}
\begin{gathered}
    \vec{x} =  \left|\sum_{k=1}^{n_c}{\vec{S}^k}^{*} \mathcal{F}^{-1}(\vec{y}^k)\right|.
    \label{eq:sense}
\end{gathered}
\end{equation}

\noindent
In accelerated acquisitions, the fully-sampled $k$-space is subsampled using a subsampling operator $\mat{U}_{\mat{M}}$, which selectively retains pixels present in the sampling set $\mat{M}$ and sets others to zero:

\begin{equation}
    \mat{U}_\mat{M}(\vec{y})_j = \vec{y}_j \cdot \mathbb{1}_\mat{M}(j) = \Bigg\{\begin{array}{lr}
        \vec{y}_j, & \text{if } j\in\mat{M}\\
        0, & \text{if } j \notin \mat{M}.
        \end{array}
\end{equation} 

\noindent
The forward problem of the accelerated acquisition is described by:

\begin{equation}
    \Tilde{\vec{y}}_\vec{M} = \mathcal{A}_{\vec{M}, \mat{S}} (\vec{x}) + \vec{e},
    \label{eq:forward_problem}
\end{equation}

\noindent
where $\vec{e}\in\mathbb{C}^{n\times n_c}$ represents measurement noise and $\mathcal{A}_{\vec{M}, \mat{S}}:\mathbb{C}^{n}\rightarrow\mathbb{C}^{n\times n_c}$ denotes the forward operator which maps an image to individual coil images using the coil sensitivity maps $\mat{S}$ through the expand operator $\mathcal{E}_\mat{S}$, applies the FFT, denoted by $\mathcal{F}$, and applies subsampling via $\mat{U}_{\mat{M}}$:

\begin{equation}
        \mathcal{A}_{\vec{M}, \mat{S}} (\vec{x}) = \mat{U}_{\mat{M}} \circ \mathcal{F} \circ \mathcal{E}_\mat{S}(\vec{x}) = \Big\{\mat{U}_{\mat{M}} \circ \mathcal{F} (\mat{S}^k \vec{x}) \Big \}_{k=1}^{n_c}.
        \label{eq:forward_op}
\end{equation}

\noindent
The adjoint operator of $\mathcal{A}_{\vec{M}, \mat{S}}$, denoted by $\mathcal{A}^{*}_{\vec{M}, \mat{S}}:\mathbb{C}^{n\times n_c}\rightarrow\mathbb{C}^{n}$, subsamples the input multi-coil data via $\mat{U}_{\mat{M}}$, maps them to the image domain using $\mathcal{F}^{-1}$, and reduces them to a single image using $\mat{S}$ via the reduce operator $\mathcal{R}_\mat{S}$:

\begin{equation}
     \mathcal{A}^{*}_{\vec{M}, \mat{S}}(\vec{y}) = \mathcal{R}_\mat{S}\circ \mathcal{F}^{-1} \circ \mat{U}_{\mat{M}} (\vec{y}) =\\
     \sum_{k=1}^{n_c} {\vec{S}^k}^{*} \mathcal{F}^{-1} \left(\mat{U}_{\mat{M}} (\vec{y}^k) \right).
     \label{eq:backward_op}
\end{equation}

\subsection{MRI Reconstruction}
\label{sec:subsec3.2}
Typically, the process of recovering an image from the subsampled $k$-space measurements $\Tilde{\vec{y}}_{\mat{M}}$ is formulated as a regularized least squares optimization problem:

\begin{equation}
    \vec{x}^{*} =  \arg\min_{\vec{x}^{'}}\frac{1}{2}\left|\left| \mathcal{A}_{\vec{M}, \mat{S}}(\vec{x}^{'}) - \Tilde{\vec{y}}_{\mat{M}}\right|\right|_2^2 + \mathcal{G}(\vec{x}^{'}),
\label{eq:variational_problem}
\end{equation}

\noindent
where $\mathcal{G}: \mathbb{C}^{n}\rightarrow\mathbb{R}$ represents an arbitrary regularization functional that incorporates prior reconstruction information. \Equation{variational_problem} lacks a closed-form solution, and a solution can only be obtained numerically.  


\subsection{MRI Reconstruction with Supervised Learning}
\label{sec:subsec3.3}
In supervised learning settings, fully-sampled $k$-space datasets are assumed to be available. Let $\mathcal{D}^{\text{SL}}= \left\{\vec{y}^{(i)}\right\}_{i=1}^{N}$ represent such a dataset, which is retrospectively subsampled during training: $\Tilde{\vec{y}}_{\mat{M}_{i}}^{(i)} = \mat{U}_{\mat{M}_{i}} (\vec{y}^{(i)})$,  and let $f_{\boldsymbol{\psi}}$ denote a DL-based reconstruction network with parameters $\boldsymbol{\psi}$.  Note that the architecture of $f_{\boldsymbol{\psi}}$ can be configured to output image reconstructions, $k$-space data, or both, but here we assume that both input and output lie in the image domain.

Throughout the paper, we use the notation $\tilde \cdot$ for subsampled measurements and $\hat \cdot$ for model predictions. The objective in SL-based MRI reconstruction is to minimize the discrepancy between the predicted and the fully-sampled $k$-spaces:
 
\begin{equation}
\begin{gathered}
    \boldsymbol{\psi}^{*} = \arg\min_{\boldsymbol{\psi}} \frac{1}{N}\sum_{i=1}^{N} \mathcal{L}_{K}\left(\vec{y}^{(i)},\hat{\vec{y}}^{(i)}\right), \\
    \hat{\vec{y}}^{(i)} = \text{DC}_{\mat{M}_{i}} \left(\Tilde{\vec{y}}_{\mat{M}_{i}}^{(i)}, \hat{\vec{y}}_{\mat{M}_{i}}^{(i)} \right),\quad
    \hat{\vec{y}}_{\mat{M}_{i}}^{(i)} = \mathcal{F} \circ \mathcal{E}_{\mat{S}}  \circ f_{\boldsymbol{\psi}}  \left( \Tilde{\Vec{x}}_{\mat{M}_{i}}^{(i)} \right),\quad \Tilde{\Vec{x}}_{\mat{M}_{i}}^{(i)} = \mathcal{A}^{*}_{\vec{M}_{i}, \mat{S}} \big(\vec{y}^{(i)}\big),
\end{gathered}
\end{equation}

or the discrepancy between the predicted image and the reconstructed fully-sampled (ground truth) image:

\begin{equation}
\begin{gathered}
    \boldsymbol{\psi}^{*} = \arg\min_{\boldsymbol{\psi}} \frac{1}{N}\sum_{i=1}^{N} \mathcal{L}_{\text{I}} \left( \vec{x}^{(i)}, \hat{\vec{x}}^{(i)} \right), \\
    \vec{x}^{(i)} =   \text{RSS} \circ \mathcal{F}^{-1} \left( \vec{y}^{(i)}  \right),\quad \hat{\vec{x}}^{(i)} =  \left| f_{\boldsymbol{\psi}} \big( \Tilde{\Vec{x}}_{\mat{M}_{i}}^{(i)} \big) \right |,
\end{gathered}
\end{equation}

where $\mathcal{L}_{K}$ and $\mathcal{L}_{\text{I}}$ denote arbitrary loss functions computed in the $k$-space and image domains, respectively. The operator DC$_\mat{M}$ denotes the data consistency operator, which enforces consistency between the available and predicted measurements and is defined as:

\begin{equation}
    \text{DC}_{\mat{M}}(\vec{w}_1, \, \vec{w}_2) = \mat{U}_{\mat{M}}(\vec{w}_1) + \mat{U}_{\mat{M}^{c}}(\vec{w}_2).
\end{equation}

\noindent
For unseen data $\Tilde{\vec{y}}^{\text{inf}}_{\mat{M}}$ a prediction is estimated as:

\begin{equation}
    \hat{\vec{x}} =  \left| f_{\boldsymbol{\psi}^{*}} \big( \Tilde{\vec{x}}^{\text{inf}}_{\mat{M}} \big) \right|, \quad \Tilde{\vec{x}}^{\text{inf}}_{\mat{M}} =  \mathcal{R}_{\mat{S}} \circ \mathcal{F}^{-1} \left( \Tilde{\vec{y}}^{\text{inf}}_{\mat{M}} \right) .
    \label{eq:sl_pred}
\end{equation}

\subsection{MRI Reconstruction with Self-supervised Learning}
\label{sec:subsec3.4}

In situations where fully-sampled $k$-space data is not available, DL models can still be trained using self-supervised learning.  Let $\mathcal{D}^{\text{SSL}} = \big\{\Tilde{\vec{y}}_{\mat{M}_{i}}^{(i)}\big\}_{i=1}^{N}$ be a dataset containing subsampled acquisitions with each instance $\Tilde{\vec{y}}_{\mat{M}_{i}}^{(i)}$ being sampled in a set $\mat{M}_{i}$. To train a reconstruction network under SSL settings, the acquired subsampled measurements are partitioned (SSDU \cite{Yaman2020}). Specifically, for each sample $\Tilde{\vec{y}}_{\mat{M}{i}}^{(i)}$, partitioning is performed by splitting the sampling set $\mat{M}_i$ into two disjoint subsets, $\mat{\Theta}_i$ and $\mat{\Lambda}_i$, and then projecting $\Tilde{\vec{y}}_{\mat{M}_{i}}^{(i)}$ onto both:

\begin{equation}
\begin{gathered}
    \mat{\Theta}_i \cup \mat{\Lambda}_i = \mat{M}_i, \quad \mat{\Theta}_i \cap \mat{\Lambda}_i = \emptyset, \\
    \Tilde{\vec{y}}_{\mat{\Theta}_{i}}^{(i)} = \mat{U}_{\mat{\Theta}_i}(\Tilde{\vec{y}}_{\mat{M}_{i}}^{(i)}),\quad \Tilde{\vec{y}}_{\mat{\Lambda}_{i}}^{(i)} = \mat{U}_{\mat{\Lambda}_i}(\Tilde{\vec{y}}_{\mat{M}_{i}}^{(i)}).
    \label{eq:ssl_partitioning}
\end{gathered}
\end{equation}

\noindent
Subsequently, one partition is used as input to the reconstruction network, while the other serves as the target. Therefore, the objective loss function is formulated in the $k$-space domain as follows:

\begin{equation}
\begin{gathered}
    \boldsymbol{\psi}^{*} = \arg\min_{\boldsymbol{\psi}} \frac{1}{N}\sum_{i=1}^{N} \mathcal{L}_{\textit{K}}\Bigg(\Tilde{\vec{y}}_{\mat{\Theta}_{i}}^{(i)},  \hat{\vec{y}}_{\mat{\Theta}_{i} \mat{\Lambda}_{i} }^{(i)}
      \Bigg),\\
     \hat{\vec{y}}_{\mat{\Theta}_{i} \mat{\Lambda}_{i} }^{(i)} = 
     \mat{U}_{\mat{\Theta}_i} \left(\text{DC}_{\boldsymbol{\Lambda}_{i}}\Big(\Tilde{\vec{y}}_{\mat{\Lambda}_{i}}^{(i)},   \hat{\Vec{y}}_{\mat{\Lambda}_{i}}^{(i)} \Big) \right) \quad 
     \hat{\Vec{y}}_{\mat{\Lambda}_{i}}^{(i)} = \mathcal{F} \circ \mathcal{E}_\mat{S} \circ  f_{\boldsymbol{\psi}}  ( \Tilde{\Vec{x}}_{\mat{\Lambda}_{i}}^{(i)} ) , \quad  \Tilde{\Vec{x}}_{\mat{\Lambda}_{i}}^{(i)} = \mathcal{A}^{*}_{\mat{\Lambda}_{i}, \mat{S}} \big(\Tilde{\vec{y}}_{\mat{M}_{i}}^{(i)}\big).
\end{gathered}
\end{equation}

\noindent
The loss can equivalently be computed in the image domain as follows:

\begin{equation}
\begin{gathered}
    \boldsymbol{\psi}^{*} = \arg\min_{\boldsymbol{\psi}} \frac{1}{N}\sum_{i=1}^{N} \mathcal{L}_{\text{I}} \left( \Tilde{\Vec{x}}_{\mat{\Theta}_{i}}^{(i)} , \hat{\vec{x}}^{(i)} \right), \\
    \Tilde{\Vec{x}}_{\mat{\Theta}_{i}}^{(i)} =   \text{RSS} \circ \mathcal{F}^{-1} \left( \Tilde{\vec{y}}_{\mat{\Theta}_{i}}^{(i)}  \right), \quad 
    \hat{\vec{x}}^{(i)} = \left|\mathcal{R}_{\mat{S}} \circ \mathcal{F}^{-1} \left( \hat{\vec{y}}_{\mat{\Theta}_{i} \mat{\Lambda}_{i} }^{(i)} \right)\right|.
\end{gathered}
\end{equation}

While most SSL-based MRI reconstruction methods rely on loss calculations in the frequency domain \cite{Yaman2020,millard2023,Hu2021,hu2024spicerselfsupervisedlearningmri}, some studies have explored dual-domain losses \cite{ZHOU2022102538,yan2023dc}. 

\noindent
For unseen data $\Tilde{\vec{y}}^{\text{inf}}_{\mat{M}}$ in SSL settings, a prediction is estimated as outlined below:
\begin{equation}
\begin{gathered}
    \hat{\vec{x}} = \left | \mathcal{R}_{\mat{S}} \circ \mathcal{F}^{-1} \circ \text{DC}_{\mat{M}} \left (\Tilde{\vec{y}}^{\text{inf}}_{\mat{M}}, \hat{\vec{y}}^{\text{inf}}_{\mat{M}}  \right ) \right |, \\ \hat{\vec{y}}^{\text{inf}}_{\mat{M}} = \mathcal{F} \circ \mathcal{E}_{\mat{S}} \circ  f_{\boldsymbol{\psi}^{*}} \big( \Tilde{\vec{x}}^{\text{inf}}_{\mat{M}} \big), \quad
    \Tilde{\vec{x}}^{\text{inf}}_{\mat{M}}  = \mathcal{R}_{\mat{S}} \circ \mathcal{F}^{-1} \left( \Tilde{\vec{y}}^{\text{inf}}_{\mat{M}} \right).
    \label{eq:ssl_pred}
\end{gathered}
\end{equation}

%

\subsection{Joint Supervised and Self-supervised Learning}
\label{sec:subsec3.5}

In this section, we present the Joint Supervised and Self-supervised Learning method. JSSL is a novel approach designed to train DL-based MRI reconstruction models in scenarios where reference data is unavailable in the target organ domain. JSSL combines elements of both supervised and self-supervised learning, employing self-supervised learning using subsampled measurements from the target domain(s) (target dataset(s)) and supervised learning with one or more datasets containing fully-sampled acquisitions from other organ domains (proxy datasets). The rationale behind JSSL is to harness the knowledge transferred from the proxy datasets, potentially improving upon conventional self-supervised methods that rely solely on the (subsampled) target dataset(s) for learning. \Figure{full_overview} illustrates the end-to-end JSSL pipeline. 

\subsection{JSSL Training Framework}
\label{sec:subsec3.6}
To implement JSSL, we construct the overall loss function with two components: one for supervised learning and another for self-supervised learning. For simplicity we assume a single target and a single proxy dataset.

\paragraph{Supervised Learning Loss}
The SL loss is calculated on the proxy dataset, which contains fully-sampled $k$-space data. It is formulated as follows:

\begin{equation}
\begin{gathered}
    {\mathcal{L}}_{\boldsymbol{\psi}}^{\text{SL}}  :=  {\mathcal{L}_{\text{I}}}_{\boldsymbol{\psi}}^{\text{SL}} +  {\mathcal{L}_{K}}_{\boldsymbol{\psi}}^{\text{SL}} =    \frac{1}{N_{\text{prx}}} \sum_{i=1}^{N_{\text{prx}}} 
    \left[ \mathcal{L}_{\text{I}} \left( \vec{x}^{\text{prx}, (i)}, \hat{\vec{x}}^{\text{prx}, (i)} \right) \,
    + \, \mathcal{L}_{K}\left(\vec{y}^{\text{prx}, (i)}, \hat{\vec{y}}^{\text{prx}, (i)}\right) \right].
\label{eq:jssl_sl_loss}
\end{gathered}
\end{equation}

\noindent
Here, $\vec{x}^{\text{prx}, (i)}$ , $\hat{\vec{x}}^{\text{prx}, (i)}$ represent the ground truth and predicted images, respectively, for the $i$-th sample in the proxy dataset, while $\vec{y}^{\text{prx}, (i)}$, $\hat{\vec{y}}^{\text{prx}, (i)}$ represent the fully-sampled and predicted $k$-spaces, respectively, as defined in \Section{subsec3.3}.

\paragraph{Self-supervised Learning Loss}
The SSL loss is calculated using the target dataset, which consists of subsampled $k$-space data without ground truth.  Motivated by other SSL-based methods \cite{zhou2022dsformer,ZHOU2022102538} which established improved performance when using dual-domain loss, we calculate the SSL loss in both the image and $k$-space domains as follows:

\begin{equation}
\begin{gathered}
    {\mathcal{L}}_{\boldsymbol{\psi}}^{\text{SSL}}  :=  {\mathcal{L}_{\text{I}}}_{\boldsymbol{\psi}}^{\text{SSL}} +  {\mathcal{L}_{K}}_{\boldsymbol{\psi}}^{\text{SSL}} =  \frac{1}{N_{\text{tar}}} 
    \sum_{i=1}^{N_{\text{tar}}} \Big[\mathcal{L}_{K}\left(\Tilde{\vec{y}}_{\mat{\Theta}_{i}}^{\text{tar}, (i)}, \hat{\vec{y}}_{\mat{\Theta}_{i} \mat{\Lambda}_{i}}^{\text{tar}, (i)}\right) + \mathcal{L}_{\text{I}} \left( \Tilde{\Vec{x}}_{\mat{\Theta}_{i}}^{\text{tar}, (i)}, \hat{\vec{x}}^{\text{tar}, (i)} \right) \Big],\\
     \Tilde{\Vec{x}}_{\mat{\Theta}_{i}}^{\text{tar}, (i)} = \text{RSS} \circ \mathcal{F}^{-1} ( \Tilde{\vec{y}}_{\mat{\Theta}_{i}}^{\text{tar}, (i)}  ), \quad
     \hat{\vec{x}}^{\text{tar}, (i)} = \left| f_{\boldsymbol{\psi}} ( \Tilde{\Vec{x}}_{\mat{\Lambda}_{i}}^{\text{tar}, (i)} ) \right |,
\label{eq:jssl_ssl_loss}
\end{gathered}
\end{equation}

\noindent
where, $\Tilde{\Vec{x}}_{\mat{\Lambda}_{i}}^{\text{tar}, (i)}$, $\Tilde{\vec{y}}_{\mat{\Theta}_{i}}^{\text{tar}, (i)}$, $\hat{\vec{y}}_{\mat{\Theta}_{i} \mat{\Lambda}_{i}}^{\text{tar}, (i)}$ are as defined in \Section{subsec3.4}. 


\begin{figure*}
    \centering
    \includegraphics[width=0.88\textwidth]{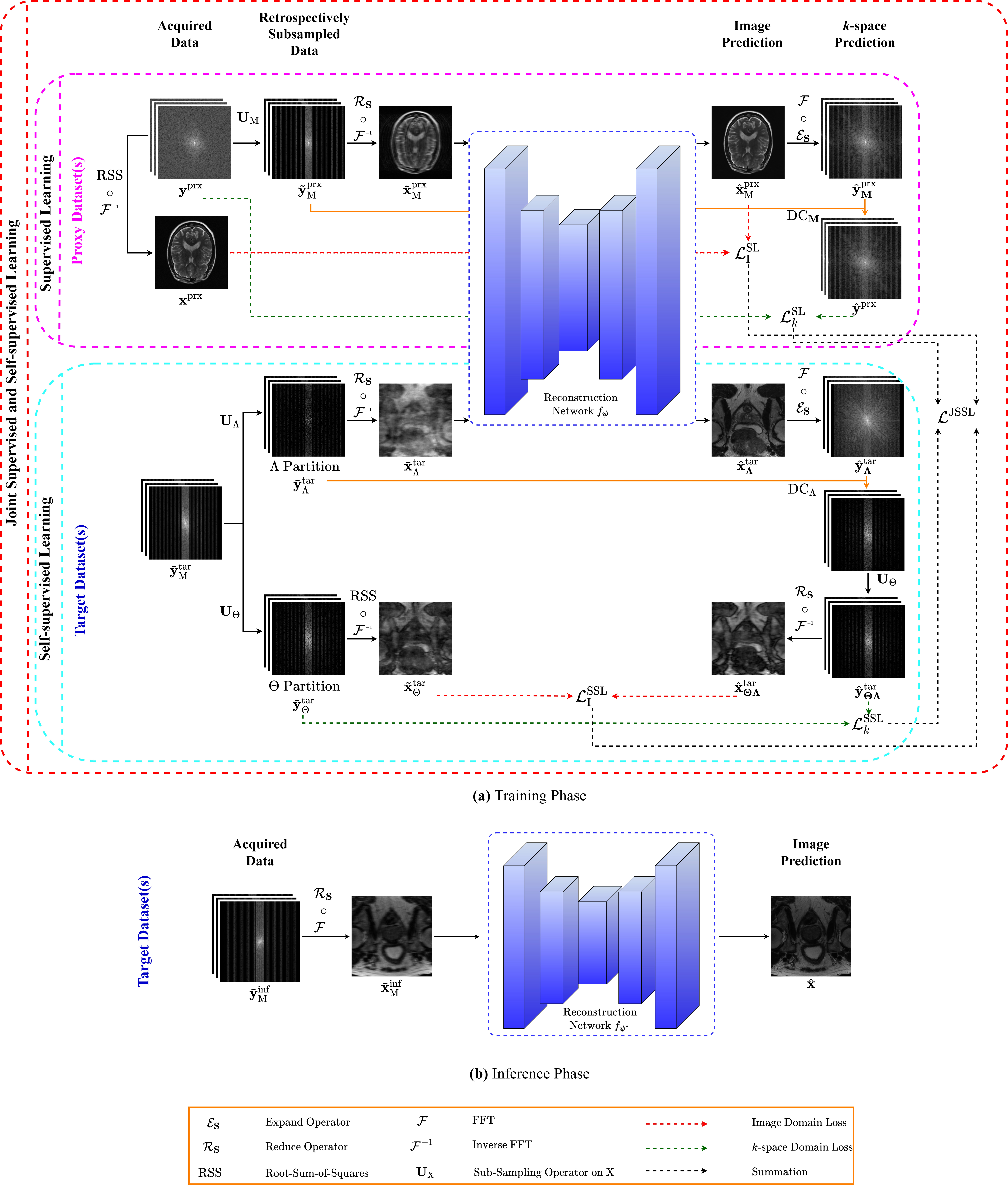}
    \caption{ \textbf{(a)} The training process for the proposed JSSL method is divided into two phases: (1) Supervised Learning using fully-sampled $k$-space data from proxy datasets. During this phase, the model is trained to predict fully-sampled data from retrospectively subsampled proxy data. (2) Self-supervised Learning utilizing subsampled $k$-space data from the target dataset, partitioned into two disjoint subsets. The model takes one subset as input and aims to predict the other. Loss functions are defined for both SL and SSL settings in both $k$-space and image domains and the JSSL loss comprises all these components. The model is jointly trained with both supervised and self-supervised loss functions  to enhance MRI reconstruction in the target domain. In all phases, sensitivity maps $\mat{S}$ are learned using the autocalibration signal (center of $k$-space) from available measurements. A U-Net, functioning as a Sensitivity Map Estimator module (omitted in the diagram)., is trained end-to-end with the reconstruction network. \textbf{(b)} In the inference phase, the trained network predicts the underlying ground truth image from the target dataset based on input subsampled data.}
    \label{fig:full_overview}
\end{figure*}

\paragraph{JSSL Loss:}
The JSSL loss is a fundamental component of our approach, defined as the combination of the SL and SSL losses: $\mathcal{L}_{\boldsymbol{\psi}}^{\text{JSSL}} := {\mathcal{L}}_{\boldsymbol{\psi}}^{\text{SL}} + {\mathcal{L}}_{\boldsymbol{\psi}}^{\text{SSL}}$ and the model's parameters are updated during training such that $\boldsymbol{\psi}^{*} = \arg\min_{\boldsymbol{\psi}} \mathcal{L}_{\boldsymbol{\psi}}^{\text{JSSL}}$.

\subsection{JSSL at Inference}
\label{sec:subsec3.7}

During the inference phase, for subsampled $k$-space data $\Tilde{\vec{y}}^{\text{inf}}_{\mat{M}}$ the trained JSSL reconstruction model $f_{\boldsymbol{\psi}^{*}}$ estimates the underlying image as follows:

\begin{equation}
    \hat{\vec{x}} =  \left | f_{\boldsymbol{\psi}^{*}} \big( \Tilde{\vec{x}}^{\text{inf}}_{\mat{M}} \big)\right |, \quad \Tilde{\vec{x}}^{\text{inf}}_{\mat{M}}  =  \mathcal{R}_{\mat{S}} \circ \mathcal{F}^{-1} \left( \Tilde{\vec{y}}^{\text{inf}}_{\mat{M}} \right),
    \label{eq:jssl_inference}
\end{equation}
where $\Tilde{\vec{y}}^{\text{inf}}_{\mat{M}}$ denotes the subsampled $k$-space data.

\subsection{JSSL: A Theoretical Motivation}
\label{sec:subsec3.8}

The core concept behind JSSL is to leverage both supervised and self-supervised learning to enhance MRI reconstruction of a target dataset, even when the parameters optimized on supervised proxy tasks may not be the most optimal. We hypothesize that introducing a supervised proxy task serves as a form of regularization, reducing the variance of our estimators due to the proxy supervised training on a `less noisy' task. We illustrate this intuition with two simplified examples in Proposition \ref{prop:prop1} (estimating means of distributions) and Proposition \ref{prop:prop2} (linear regression), where we assume two distributions - one that we wish to estimate, but we cannot obtain sufficient samples from, and a proxy distribution that is directly accessible. We demonstrate that drawing samples from both distributions (or using only the proxy distribution) can reduce our estimator's variance and risk.

\begin{prop}
    Consider two distributions $p_i, \, i=1,2$ with means and variances ${\mu}_i, {\sigma}_i, \, i=1,2$, with unknown ${\mu}_1$, and ${\mu}_1 \neq {\mu}_2$. Then if $(\mu_1 - \mu_2)^2 < c \frac{{\sigma}_1^2} N$ for some $c \in (0, 1)$ and $N \in \mathbb{Z}^{+}$, then $\tilde{{x}} = \frac{1}{N+K} \sum_{i=1}^{N+K} {x}_i$ is a lower-variance estimator of $\mu_1$ compared to $\overline{{x}} = \frac{1}{N} \sum_{i=1}^N {x}_i$, where $\left\{{x}^{(i)} \sim p_1\right\}_{i=1}^{N}$ and $\left\{{x}^{(N+i)} \sim p_2 \right\}_{i=1}^{K}$ for a choice of a large $K \in \mathbb{Z}^{+} $.
\label{prop:prop1}
\end{prop}
\begin{proof}
    See Supporting Information Appendix \ref{sec:ap1}.
\end{proof}

\begin{prop}
\label{prop:prop2}
Let $\bm x \sim \mathcal N(\bm 0, \sigma^2 \bm I_{p})$ be $\mathbb R^p$-valued isotropic Gaussian random vector and $y, \tilde y$ be random variables with $p(y|\bm x) = \mathcal N(y| \bm w^T \bm x, \varepsilon^2)$ and $p(\tilde y|\bm x) = \mathcal N(\tilde y| \bm{\tilde w}^T \bm x, \tilde \varepsilon^2)$ for some $\bm w, \bm{\tilde w} \in \mathbb R^p$. Let $\mathcal T = \{(\bm x_1, \tilde y_1),\dots,  (\bm x_K, \tilde y_K)\}$ be a training data set with $K > p$ and consider a maximum likelihood estimator $\widehat y(\bm x; \mathcal T)$. Then the following holds:
\begin{enumerate}[ noitemsep, topsep=0pt]
\item $\mathrm{Bias}_{\mathcal T}[\widehat y(\bm x; \mathcal T)] = (\bm{\tilde w}^T - \bm w^T) \bm x.$
\item $\mathrm{Var}_{\mathcal T}[\widehat y(\bm x; \mathcal T)] =  \frac{\tilde \varepsilon^2}{\sigma^2 K} \| \bm x \|_2^2$.
\item $\mathbb E_{(\bm x, y)} [\widehat y(\bm x; \mathcal T) - y]^2 \leq  p \sigma^2 \| \bm {\tilde w} - \bm w \|_2^2  + \frac{p \tilde \varepsilon^2}{K}  + \varepsilon^2$
\end{enumerate}
\end{prop}

\begin{proof}
    See Supporting Information Appendix \ref{sec:ap1}.
\end{proof}

Propositions \ref{prop:prop1} and \ref{prop:prop2} imply that leveraging a large number of samples from the proxy distribution ($K \to \infty$) can lead to a significant reduction in the variance of estimators trained under both supervised and self-supervised learning paradigms. Moreover, it highlights how the introduction of bias through supervised learning can be a strategic trade-off to lower variance. Additionally, Proposition \ref{prop:prop2} sheds light on how the risk associated with our estimator can be influenced by the degree of similarity between the target and proxy distributions.

\subsection{Reconstruction Network}
\label{sec:subsec3.9}

To circumvent the need for numerical optimization, DL methods have been deployed, enabling models to learn the reconstruction process directly from data. A multitude of DL approaches with varying configurations concerning the domain of operation, architectural design, and physics-guided unrolled considerations have been developed  \cite{Liang2020, bioengineering10091012}. These methods utilize the forward and backward operators in \eqref{eq:forward_op} and \eqref{eq:backward_op} and aim to solve \eqref{eq:variational_problem} iteratively. They have demonstrated significant reconstruction capabilities in both supervised and self-supervised approaches \cite{Yaman2020, Korkmaz2023, Hammernik2017, Yiasemis_2022_CVPR}.

In this study, we utilized the Variable Splitting Half-quadratic ADMM algorithm for Reconstruction of Inverse Problems (vSHARP) \cite{yiasemis2023vsharp}, a physics-guided deep learning method unrolled across iterations, previously applied in accelerated brain, prostate and cardiac MRI reconstruction \cite{yiasemis2023vsharp, yiasemis2023deep}. The vSHARP algorithm uses the half-quadratic variable splitting technique and the Alternating Direction Method of Multipliers (ADMM) over $T$ iterations to iteratively solve the optimization problem in \eqref{eq:variational_problem}. Each iteration includes three key steps: (a) denoising the auxiliary variable introduced by variable splitting using a deep learning-based denoiser, (b) updating the target image data consistency via an unrolled gradient descent scheme over $T_{\vec{x}}$ iterations, and (c) updating the Lagrange multipliers. Additional details can be found in the Supporting Information Appendix \ref{sec:ap2}.

While vSHARP serves as the primary model in our experiments, it is important to emphasize that training frameworks are model-agnostic. In \Section{subsubsec3.11.5}, we introduce two additional architectures to highlight this versatility.

\subsection{Coil Sensitivity Prediction}
\label{sec:subsec3.10}

The initial approximation of coil sensitivity maps is derived from the autocalibration signal (ACS), specifically the center of the $k$-space\cite{https://doi.org/10.1002/mrm.10087}. While SSL-based approaches such as \cite{Yaman2020,desai2021noise2recon,ZHOU2022102538} use this initial approximation or employ expensive algorithms like Espirit \cite{Uecker2013}, our JSSL approach takes this initial estimation and feeds it as input to a Sensitivity Map Estimator (SME) similarly to \cite{millard2023,pmlr-v227-zhang24a,hu2024spicerselfsupervisedlearningmri}. The SME is a DL-based model designed to enhance and refine the sensitivity profiles and it is trained end-to-end in conjunction with the reconstruction model. Note that we integrate a SME module in all our experiments in \Section{subsec3.11}.


\subsection{Experiments}
\label{sec:subsec3.11}
All experiments were conducted using the Deep Image Reconstruction Toolkit (DIRECT) \cite{DIRECTTOOLKIT}. The complete codebase, including data loading, processing functions, and models, is available at \url{https://github.com/NKI-AI/direct}.

\subsubsection{Datasets}
\label{sec:subsubsec3.11.1}
We utilized several datasets that contain fully-sampled multi-coil $k$-space data: the fastMRI brain \cite{zbontar2019fastmri}, fastMRI knee \cite{zbontar2019fastmri}, fastMRI prostate T2 \cite{tibrewala2023fastmri}, and CMRxRecon challenge 2023 cardiac cine MRI \cite{cmrxrecon,cmrxrecondataset} datasets. The characteristics of these datasets and the parameters for data splitting are summarized in Supporting Information Table \ref{tab:S1}.

To evaluate the performance of JSSL, we selected two different combinations of target and proxy datasets for our experiments:

\begin{enumerate}[leftmargin=*,label=(\Alph*)]
    \item \textbf{Target:} Prostate dataset. \textbf{Proxy:} Brain and knee datasets.
    \item \textbf{Target:} Cardiac dataset. \textbf{Proxy:} Brain, knee, and prostate datasets.
\end{enumerate}

During training, the fully-sampled data from the target datasets were retrospectively subsampled. Fully-sampled target data were used exclusively for evaluation during inference. Similarly, the proxy data were retrospectively subsampled for training purposes, with fully-sampled measurements used for loss calculation.  Due to the substantial difference in the number of training samples between the proxy and target datasets, we implemented an oversampling strategy on the training proxy data, repeating each sample twice.

\subsubsection{Subsampling Schemes}
\label{sec:subsubsec3.11.2}
For our experiments, we used a random uniform Cartesian subsampling scheme for the brain and an equispaced Cartesian subsampling scheme for the knee measurements, following the corresponding publication \cite{zbontar2019fastmri}. For the prostate data, we enforced an equispaced subsampling scheme, as it's one of the easiest and fastest to implement on scanners, suitable for prostate imaging. For the cardiac dataset, we utilized the provided equispaced-like schemes by the challenge \cite{cmrxrecon}.

The subsampling process during training involved randomly selecting accelerations of $R=$ 4, 8, 12 (only for \textbf{A}), or 16 (only for \textbf{B}). Specifically, for an acceleration factor of 4, 8\% of the fully-sampled data were retained as ACS lines (center of $k$-space). Similarly, for acceleration factors of 8, 12, and 16, the corresponding percentages of ACS lines were 4\%, 3\%, and 2\% of the fully-sampled data, respectively, as  described in \cite{yiasemis2023retrospective}. During inference, our methods were tested under acceleration factors of $R =$ 2, 4, 8, 12 (only for \textbf{A}) and 16 (only for \textbf{B}), with ACS percentages of 16\%, 8\%, 4\%, 3\%, and 2\%, respectively.

\subsubsection{SSL Subsampling Partitioning}
\label{sec:subsubsec3.11.3}
During the training of any SSL-based method, including JSSL, in our experiments, the subsampled data underwent partitioning into two distinct sets, as elaborated in  \Section{subsec3.4}. To achieve this, 
$\mat{\Theta}_{i}$ was obtained  by selecting elements from ${\mat{M}_{i}}$ via a 2D Gaussian scheme with a standard deviation of 3.5 pixels. This choice is backed by literature suggesting Gaussian outperforms uniform partitioning \cite{Yaman2020}, although identifying optimal partitioning schemes is out of the scope of this work. Consequently, we set ${\mat{\Lambda}_{i}} = {\mat{M}_{i} \setminus \mat{\Theta}_{i}}$ (more information is provided in Supporting Information Appendix \ref{sec:ap2}). Furthermore, the ratio $q_i = \frac{|\mat{\Theta}_{i}|}{|\mat{M}_{i}|}$ was randomly selected between 0.3 and 0.8. An illustrative example of this is provided in Supporting Information Figure \ref{fig:S1}. Note that a $w \times w = 4 \times 4$ window in the center of the ACS region was included in each $\mat{\Lambda}_{i}$ to facilitate effective training of the SME module.

\subsubsection{Implementation \& Optimization}
\label{sec:subsubsec3.11.4}
\paragraph{Model Architecture} 

In all our experiments, we adopted vSHARP with $T=12$ optimization steps, utilizing two-dimensional U-Nets \cite{Ronneberger2015} composed of 4 scales and 32 filters (in the first scale) for $\left\{\mathcal{H}_{\boldsymbol{\theta}_{t}}\right\}_{t=0}^{T-1}$. For the data consistency step, we set $T_{\vec{x}}=10$. For the SME module we employed a 2D U-Net with 4 scales and 16 filters in the first scale.

\paragraph{Parameter Optimization}
We optimized the model parameters using the Adam optimizer \cite{kingma2017adam}, with parameters $\epsilon=10^{-8}$, $(\beta_1, \beta_2) = (0.99, 0.999)$ and initial learning rate (lr) set to 0.003. We also employed a lr scheduler which decayed the lr by a factor of 0.8 every 150,000 training iterations. Our experiments were carried out on two A6000 RTX GPUs, with a batch size of 2 slices of multi-coil $k$-space data assigned to each GPU. All models were trained to convergence.

\paragraph{Choice of Loss Function}
\label{sec:para3.7.4.3}
In all our experiments, loss was computed as detailed in \Section{subsec3.6} employing the following:
\begin{equation*}
\begin{gathered}
    {\mathcal{L}_{\text{I}}}^{\text{SL}},\, {\mathcal{L}_{\text{I}}}^{\text{SSL}} := 2 \left(1 - \text{SSIM}  + \mathcal{L}_{1}\right) + \text{HFEN}_1 + \text{HFEN}_2,\\
    {\mathcal{L}_{K}}^{\text{SL}}, \, {\mathcal{L}_{K}}^{\text{SSL}} := 2 \left( \text{NMSE} + \text{NMAE} \right).
\end{gathered}
\end{equation*}
\noindent
For brevity, we have omitted the definitions of the individual components of these loss functions. Comprehensive details can be found in the Supporting Information Appendix \ref{sec:ap2}.

\subsubsection{Training Setups Comparison}
\label{sec:subsubsec3.11.5}
Here, we present the comparisons that were conducted to evaluate JSSL for accelerated MRI reconstruction. All of our experiments were evaluated on the test sets of the target datasets, with the aim of assessing the performance of each strategy. We performed the following experiments:
\begin{enumerate}[label=(\arabic*)]
    \item SSL in the target domain.
    \item SSL in both the target and proxy domains (SSL ALL).
    \item SL in the target domain.
    \item SL in both the proxy and target domains (SL ALL).
    \item SL in proxy domains only - out-of-distribution inference (SL Proxy).
    \item JSSL (SL in proxy domain, SSL in target domain).
\end{enumerate}

Our principal objective throughout these comparisons was to examine the performance of JSSL in relation to SSL training approaches, since we are interested in scenarios where there is no access to fully-sampled data in the target domain. To demonstrate that JSSL's superiority does not solely stem from the larger dataset size (more data introduced from the proxy datasets), we also conducted experiments using all available data using a SSL strategy, incorporating both the target and proxy datasets. SL-based experiments served as reference, although naturally, the results are expected to favor SL methods when fully-sampled data is accessible in the target domain.

\paragraph{Robustness to Model Choice Experiments}
To underscore that our comparative analysis results were not dependent on specific architectural choices, we performed additional experiments with different deep MRI reconstruction networks. These included a traditional deep MRI reconstruction network, a U-Net  \cite{Ronneberger2015} operating in the image domain, as well as a state-of-the-art physics-guided network, the End-to-End Variational Network (E2EVarNet) \cite{Sriram2020}. We repeated the JSSL and SSL-based comparison experiments described in the previous sections.

Both models underwent training and evaluation on data subsampled with acceleration factors of 4, 8, and 16, with ACS ratios of 8\%, 4\%, and 2\% of the data shape. The selection of hyper-parameters for JSSL and SSL, choices for proxy and target datasets, and data splits were consistent with the comparative experiments presented in the previous sub-sections. Detailed information on the hyper-parameter selection and experimental settings for these architectures can be found in Supporting Information Table \ref{tab:S2} and Supporting Information Appendix \ref{sec:ap3}.

\subsubsection{Alternative Configurations Studies}
\label{sec:subsubsec3.11.6}
To investigate JSSL further and evaluate its performance under different settings, we examine additional configurations for the JSSL and SSL setups. In particular, we perform the following experiments on the first set of experiments (prostate data as target):
\begin{enumerate}[label=(\arabic*)]
    \item JSSL and SSL in all domains by oversampling $10$-fold the target dataset during training to balance better proxy and target data, in comparison to $2$-fold in our original experiments in \Section{subsubsec3.11.5}.
    \item JSSL and SSL using a constant partitioning ratio of $q=0.5$ instead of $q\in(0.3-0.8)$ as in  \Section{subsubsec3.11.5}.
    \item JSSL and SSL setting for the ACS window $w\times w=10 \times 10$ opposed to $w\times w=4\times 4$ in our experiments in \Section{subsubsec3.11.5}.
\end{enumerate}

\subsubsection{Statistical Testing}
\label{sec:subsubsec3.11.7}

To determine whether the top-performing method in each category (SL methods, SSL-based methods including JSSL, SSL-based methods with different configurations) significantly outperformed the others, we conducted statistical tests. Initially, we calculated the differences in performance between the best method and the other methods within each category. The Shapiro-Wilk test \cite{SHAPIRO1965} was used to assess the normality of these differences. If the differences were normally distributed ($p > \alpha $), a paired t-test was performed, alternatively the Wilcoxon signed-rank test \cite{conover1999practical} was used. In our reported results we denote with an asterisk instances which the average best method was not found to be statistically significantly better ($p > \alpha$). Note that we set $\alpha = 0.05$ as the significance level.

\subsubsection{Evaluation}
\label{sec:subsubsec3.11.8}

To assess the results of our experiments, we employed three key metrics: the Structural Similarity Index Measure (SSIM), peak Signal-to-Noise Ratio (pSNR), and normalized mean squared error (NMSE). Metrics were calculated by comparing model outputs with the RSS ground truth reconstructions, as detailed in \eqref{eq:rss}. The metric definitions were consistent with \cite{yiasemis2023retrospective}. The selection of optimal model checkpoints was based on their performance on the validation set.

\section{Results}
\label{sec:sec4}


\subsection{Training Setups Comparison}
\label{sec:subsec4.1}
\begin{figure*}[!hbt]
    \centering
    \includegraphics[width=1\textwidth]{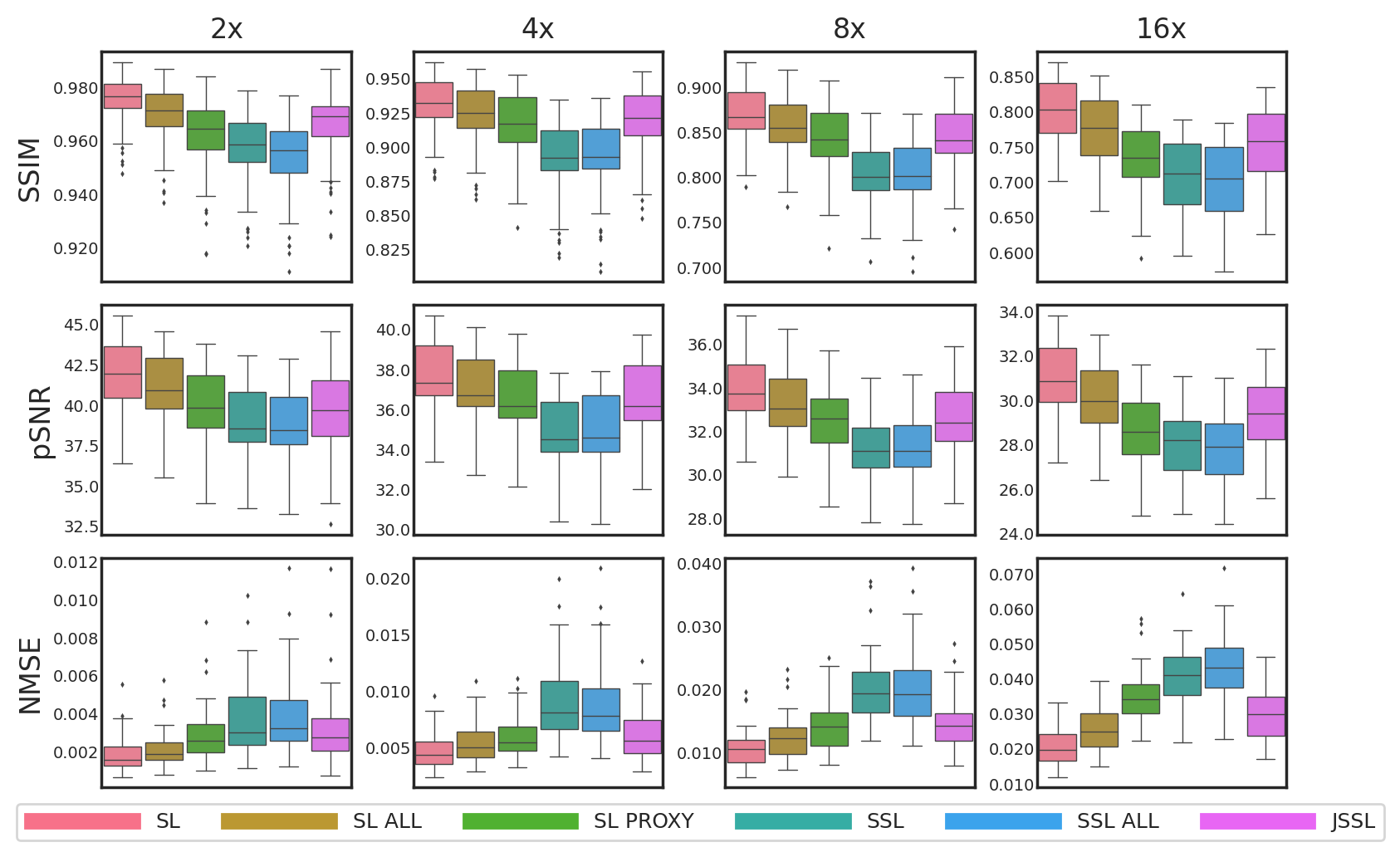}
    \caption{Evaluation results for different training setups  with the prostate as target dataset, and the brain and knee datasets as proxy datasets}
    \label{fig:metrics}
\end{figure*}

\begin{figure*}[!hbt]
    \centering
    \includegraphics[width=1\textwidth]{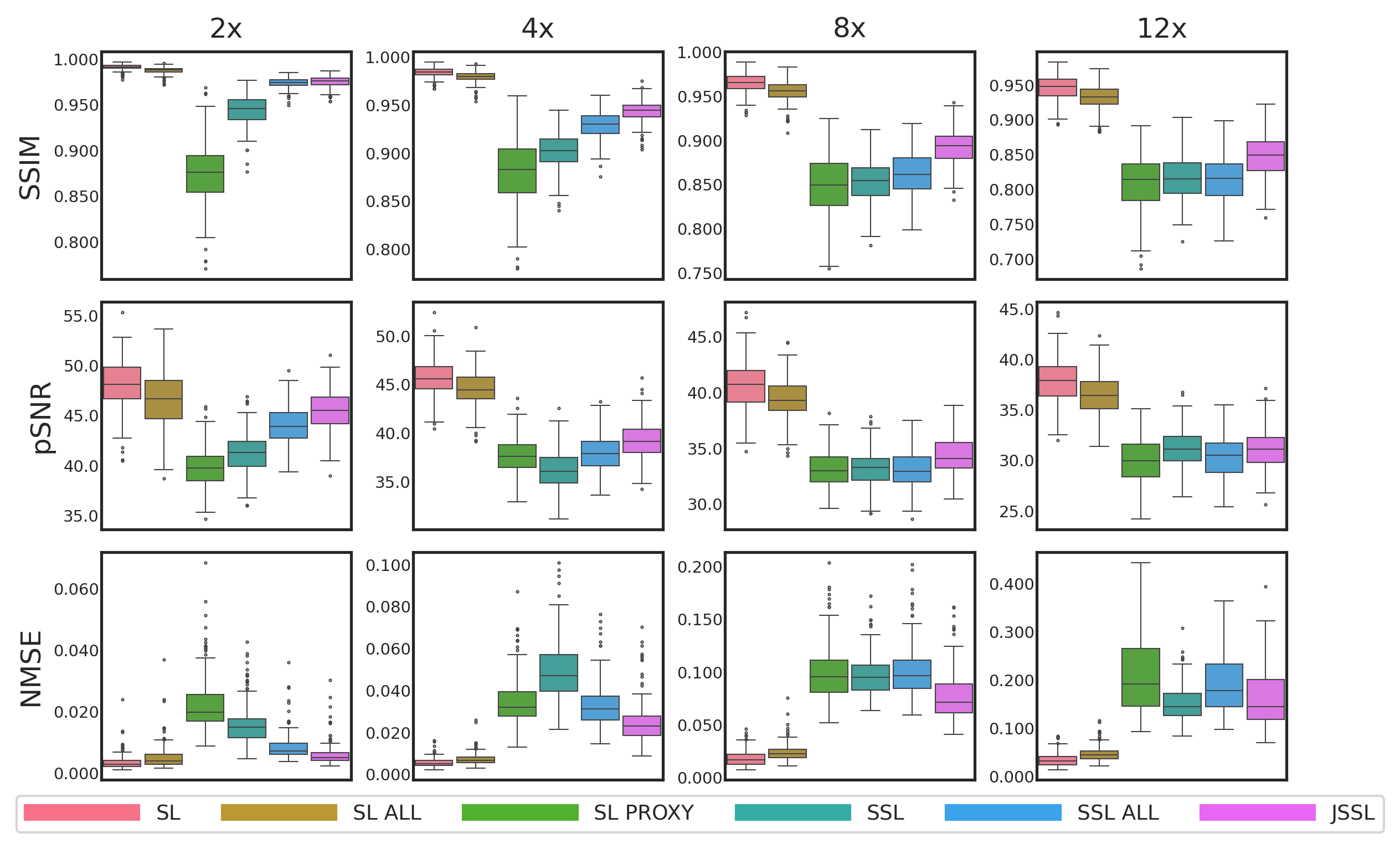}
    \caption{Evaluation results for different training setups  with the cardiac as target dataset, and the brain, knee and prostate datasets as proxy datasets}
    \label{fig:metrics-cine}
\end{figure*}


\begin{table*}[!htb]
\centering
\caption{Average metric results across different training setups with the fastMRI prostate as target dataset and the brain and knee datasets as proxy datasets. 
}
\setlength{\tabcolsep}{0.8pt}
{\renewcommand{\arraystretch}{2}
\resizebox{1\textwidth}{!}{%
\begin{tabular}{ccccccccccccc}
\hline
\multirow{2}{*}{\textbf{Setup}} &
  \multicolumn{3}{c}{\textbf{2x}} &
  \multicolumn{3}{c}{\textbf{4x}} &
  \multicolumn{3}{c}{\textbf{8x}} &
  \multicolumn{3}{c}{\textbf{16x}} \\ \cline{2-13} 
 &
  SSIM &
  pSNR &
  NMSE &
  SSIM &
  pSNR &
  NMSE &
  SSIM &
  pSNR &
  NMSE &
  SSIM &
  pSNR &
  NMSE \\ \hline
SL &
  \textbf{0.974${\scriptscriptstyle\pm 0.010}$} &
  \textbf{41.8${\scriptscriptstyle\pm 2.3}$} &
  \textbf{0.002${\scriptscriptstyle\pm 0.001}$} &
  \textbf{0.930${\scriptscriptstyle\pm 0.022}$} &
  \textbf{37.5${\scriptscriptstyle\pm 1.8}$} &
  \textbf{0.005${\scriptscriptstyle\pm 0.002}$} &
  \textbf{0.868${\scriptscriptstyle\pm 0.033}$} &
  \textbf{33.9${\scriptscriptstyle\pm 1.6}$} &
  \textbf{0.011${\scriptscriptstyle\pm 0.003}$} &
  \textbf{0.799${\scriptscriptstyle\pm 0.045}$} &
  \textbf{31.0${\scriptscriptstyle\pm 1.6}$} &
  \textbf{0.021${\scriptscriptstyle\pm 0.005}$} \\
SL ALL &
  0.969${\scriptscriptstyle\pm 0.012}$ &
  41.1${\scriptscriptstyle\pm 2.3}$ &
  0.002${\scriptscriptstyle\pm 0.001}$ &
  0.922${\scriptscriptstyle\pm 0.024}$ &
  36.9${\scriptscriptstyle\pm 1.8}$ &
  0.005${\scriptscriptstyle\pm 0.002}$ &
  0.854${\scriptscriptstyle\pm 0.035}$ &
  33.2${\scriptscriptstyle\pm 1.5}$ &
  0.013${\scriptscriptstyle\pm 0.003}$ &
  0.771${\scriptscriptstyle\pm 0.049}$ &
  30.0${\scriptscriptstyle\pm 1.6}$ &
  0.026${\scriptscriptstyle\pm 0.006}$ \\
SL PROXY &
  0.961${\scriptscriptstyle\pm 0.016}$ &
  39.8${\scriptscriptstyle\pm 2.4}$ &
  0.003${\scriptscriptstyle\pm 0.002}$ &
  0.914${\scriptscriptstyle\pm 0.026}$ &
  36.4${\scriptscriptstyle\pm 1.8}$ &
  0.006${\scriptscriptstyle\pm 0.002}$ &
  0.839${\scriptscriptstyle\pm 0.041}$ &
  32.5${\scriptscriptstyle\pm 1.7}$ &
  0.015${\scriptscriptstyle\pm 0.004}$ &
  0.733${\scriptscriptstyle\pm 0.051}$ &
  28.6${\scriptscriptstyle\pm 1.5}$ &
  0.035${\scriptscriptstyle\pm 0.008}$ \\ \hline
SSL &
  0.956${\scriptscriptstyle\pm 0.015}$ &
  38.8${\scriptscriptstyle\pm 2.6}$ &
  0.004${\scriptscriptstyle\pm 0.002}$ &
  0.891${\scriptscriptstyle\pm 0.030}$ &
  34.7${\scriptscriptstyle\pm 2.0}$ &
  0.009${\scriptscriptstyle\pm 0.003}$ &
  0.801${\scriptscriptstyle\pm 0.038}$ &
  31.1${\scriptscriptstyle\pm 1.5}$ &
  0.020${\scriptscriptstyle\pm 0.005}$ &
  0.707${\scriptscriptstyle\pm 0.050}$ &
  28.0${\scriptscriptstyle\pm 1.6}$ &
  0.041${\scriptscriptstyle\pm 0.008}$ \\
SSL ALL &
  0.953${\scriptscriptstyle\pm 0.016}$ &
  38.6${\scriptscriptstyle\pm 2.5}$ &
  0.004${\scriptscriptstyle\pm 0.002}$ &
  0.892${\scriptscriptstyle\pm 0.031}$ &
  34.8${\scriptscriptstyle\pm 2.0}$ &
  0.009${\scriptscriptstyle\pm 0.004}$ &
  0.801${\scriptscriptstyle\pm 0.041}$ &
  31.1${\scriptscriptstyle\pm 1.6}$ &
  0.020${\scriptscriptstyle\pm 0.006}$ &
  0.699${\scriptscriptstyle\pm 0.052}$ &
  27.8${\scriptscriptstyle\pm 1.6}$ &
  0.043${\scriptscriptstyle\pm 0.010}$ \\
JSSL &
  \textbf{0.965${\scriptscriptstyle\pm 0.015}$} &
  \textbf{39.5${\scriptscriptstyle\pm 2.8}$} &
  \textbf{0.003${\scriptscriptstyle\pm 0.002}$} &
  \textbf{0.918${\scriptscriptstyle\pm 0.026}$} &
  \textbf{36.4${\scriptscriptstyle\pm 1.9}$} &
  \textbf{0.006${\scriptscriptstyle\pm 0.002}$} &
  \textbf{0.842${\scriptscriptstyle\pm 0.038}$} &
  \textbf{32.5${\scriptscriptstyle\pm 1.6}$} &
  \textbf{0.015${\scriptscriptstyle\pm 0.004}$} &
  \textbf{0.752${\scriptscriptstyle\pm 0.053}$} &
  \textbf{29.3${\scriptscriptstyle\pm 1.6}$} &
  \textbf{0.030${\scriptscriptstyle\pm 0.007}$} \\ \hline
\end{tabular}%
}}

\label{tab:metrics}
\end{table*}


\begin{table*}[!htb]
\centering
\caption{Average metric results across different training setups with the CRMxRecon cardiac data as target dataset and the brain, knee and prostate datasets as proxy datasets.
}
\resizebox{\textwidth}{!}{%
\setlength{\tabcolsep}{0.8pt}
{\renewcommand{\arraystretch}{2}
\begin{tabular}{ccccccccccccc}
\hline
\multirow{2}{*}{\textbf{Setup}} &
  \multicolumn{3}{c}{\textbf{2x}} &
  \multicolumn{3}{c}{\textbf{4x}} &
  \multicolumn{3}{c}{\textbf{8x}} &
  \multicolumn{3}{c}{\textbf{12x}} \\ \cline{2-13} 
 &
  SSIM &
  pSNR &
  NMSE &
  SSIM &
  pSNR &
  NMSE &
  SSIM &
  pSNR &
  NMSE &
  SSIM &
  pSNR &
  NMSE \\ \hline
SL &
  \textbf{0.991${\scriptscriptstyle\pm 0.003}$} &
  \textbf{48.1${\scriptscriptstyle\pm 2.5}$} &
  \textbf{0.004${\scriptscriptstyle\pm 0.003}$} &
  \textbf{0.984${\scriptscriptstyle\pm 0.005}$} &
  \textbf{45.7${\scriptscriptstyle\pm 2.0}$} &
  \textbf{0.006${\scriptscriptstyle\pm 0.002}$} &
  \textbf{0.965${\scriptscriptstyle\pm 0.011}$} &
  \textbf{40.6${\scriptscriptstyle\pm 2.2}$} &
  \textbf{0.018${\scriptscriptstyle\pm 0.007}$} &
  \textbf{0.946${\scriptscriptstyle\pm 0.018}$} &
  \textbf{37.8${\scriptscriptstyle\pm 2.3}$} &
  \textbf{0.035${\scriptscriptstyle\pm 0.015}$} \\
SL ALL &
  0.987${\scriptscriptstyle\pm 0.004}$ &
  46.5${\scriptscriptstyle\pm 2.6}$ &
  0.005${\scriptscriptstyle\pm 0.004}$ &
  0.979${\scriptscriptstyle\pm 0.006}$ &
  44.5${\scriptscriptstyle\pm 1.9}$ &
  0.007${\scriptscriptstyle\pm 0.003}$ &
  0.956${\scriptscriptstyle\pm 0.012}$ &
  39.4${\scriptscriptstyle\pm 1.9}$ &
  0.024${\scriptscriptstyle\pm 0.008}$ &
  0.932${\scriptscriptstyle\pm 0.019}$ &
  36.5${\scriptscriptstyle\pm 2.0}$ &
  0.047${\scriptscriptstyle\pm 0.016}$ \\
SL PROXY &
  0.875${\scriptscriptstyle\pm 0.037}$ &
  39.8${\scriptscriptstyle\pm 2.0}$ &
  0.022${\scriptscriptstyle\pm 0.009}$ &
  0.880${\scriptscriptstyle\pm 0.035}$ &
  37.6${\scriptscriptstyle\pm 2.0}$ &
  0.036${\scriptscriptstyle\pm 0.012}$ &
  0.848${\scriptscriptstyle\pm 0.034}$ &
  33.1${\scriptscriptstyle\pm 1.7}$ &
  0.099${\scriptscriptstyle\pm 0.027}$ &
  0.810${\scriptscriptstyle\pm 0.041}$ &
  30.0${\scriptscriptstyle\pm 2.2}$ &
  0.211${\scriptscriptstyle\pm 0.079}$ \\ \hline
SSL &
  0.944${\scriptscriptstyle\pm 0.017}$ &
  41.2${\scriptscriptstyle\pm 2.1}$ &
  0.016${\scriptscriptstyle\pm 0.007}$ &
  0.902${\scriptscriptstyle\pm 0.020}$ &
  36.2${\scriptscriptstyle\pm 2.0}$ &
  0.049${\scriptscriptstyle\pm 0.014}$ &
  0.854${\scriptscriptstyle\pm 0.025}$ &
  33.2${\scriptscriptstyle\pm 1.7}$ &
  0.097${\scriptscriptstyle\pm 0.020}$ &
  0.817${\scriptscriptstyle\pm 0.032}$ &
  \textbf{31.2${\scriptscriptstyle\pm 1.9}$} &
  \textbf{0.153${\scriptscriptstyle\pm 0.038}$} \\
SSL ALL &
  0.974${\scriptscriptstyle\pm 0.006}$ &
  44.0${\scriptscriptstyle\pm 1.9}$ &
  0.009${\scriptscriptstyle\pm 0.005}$ &
  0.929${\scriptscriptstyle\pm 0.016}$ &
  37.9${\scriptscriptstyle\pm 1.9}$ &
  0.033${\scriptscriptstyle\pm 0.011}$ &
  0.862${\scriptscriptstyle\pm 0.026}$ &
  33.0${\scriptscriptstyle\pm 1.7}$ &
  0.102${\scriptscriptstyle\pm 0.026}$ &
  0.814${\scriptscriptstyle\pm 0.034}$ &
  30.3${\scriptscriptstyle\pm 2.0}$ &
  0.191${\scriptscriptstyle\pm 0.059}$ \\
JSSL &
  \textbf{0.975${\scriptscriptstyle\pm 0.007}$} &
  \textbf{45.5${\scriptscriptstyle\pm 2.0}$} &
  \textbf{0.006${\scriptscriptstyle\pm 0.004}$} &
  \textbf{0.944${\scriptscriptstyle\pm 0.013}$} &
  \textbf{39.2${\scriptscriptstyle\pm 2.0}$} &
  \textbf{0.025${\scriptscriptstyle\pm 0.010}$} &
  \textbf{0.893${\scriptscriptstyle\pm 0.022}$} &
  \textbf{34.3${\scriptscriptstyle\pm 1.8}$} &
  \textbf{0.077${\scriptscriptstyle\pm 0.023}$} &
  \textbf{0.848${\scriptscriptstyle\pm 0.032}$} &
  31.1${\scriptscriptstyle\pm 2.1}$$^{*}$ &
  0.161${\scriptscriptstyle\pm 0.059}$$^{*}$ \\ \hline
\end{tabular}%
}
}
\label{tab:metrics-cine}
\end{table*}

\begin{figure*}[!ht]
    \centering
    \includegraphics[width=\textwidth]{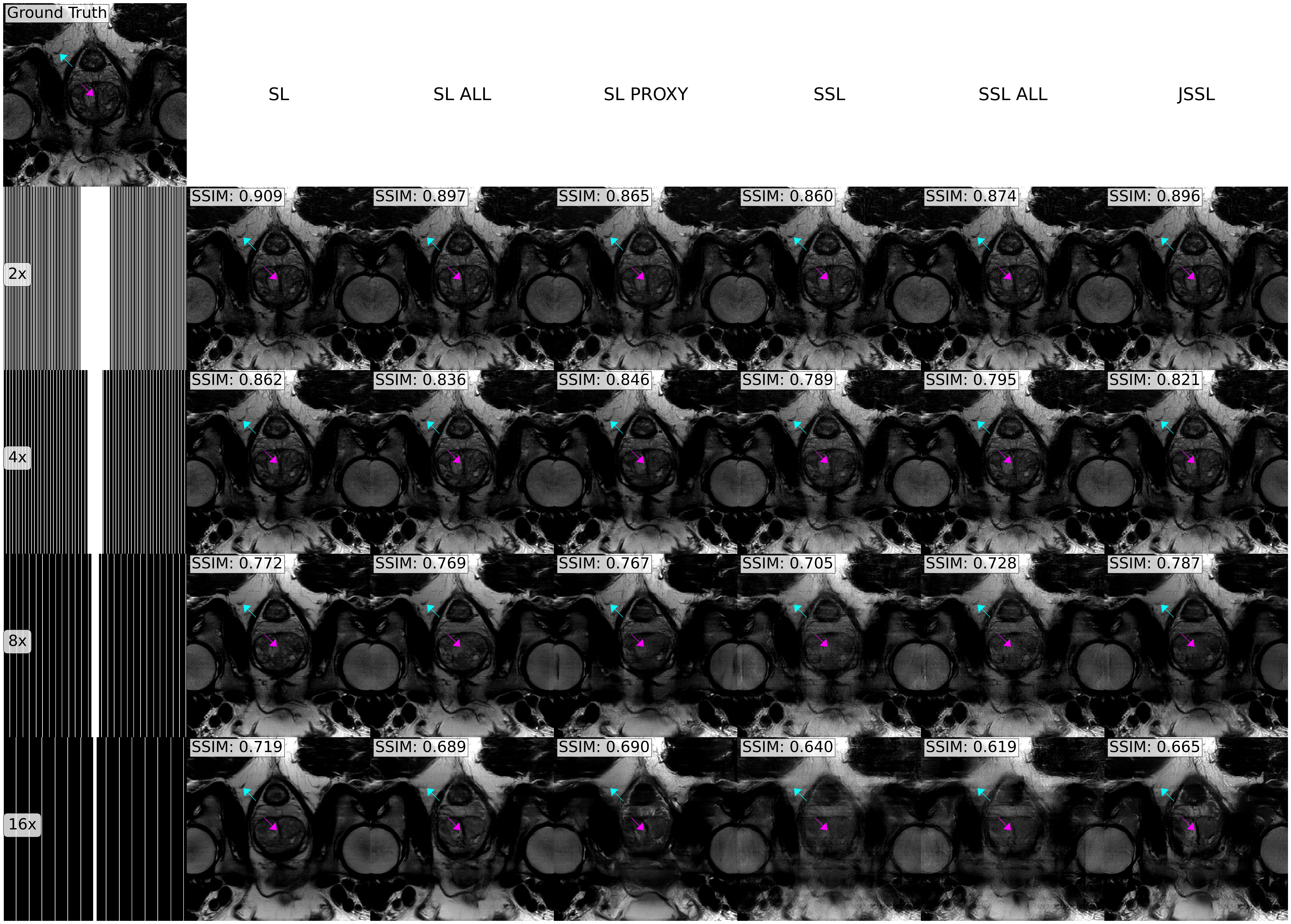}
    \caption{Example reconstructions of a slice from the prostate dataset subsampled at different acceleration factors from the test set (experiment set \textbf{A}) from each training setup.}
    \label{fig:example_figure}
\end{figure*}

\begin{figure*}[!ht]
    \centering
    \includegraphics[width=0.95\textwidth]{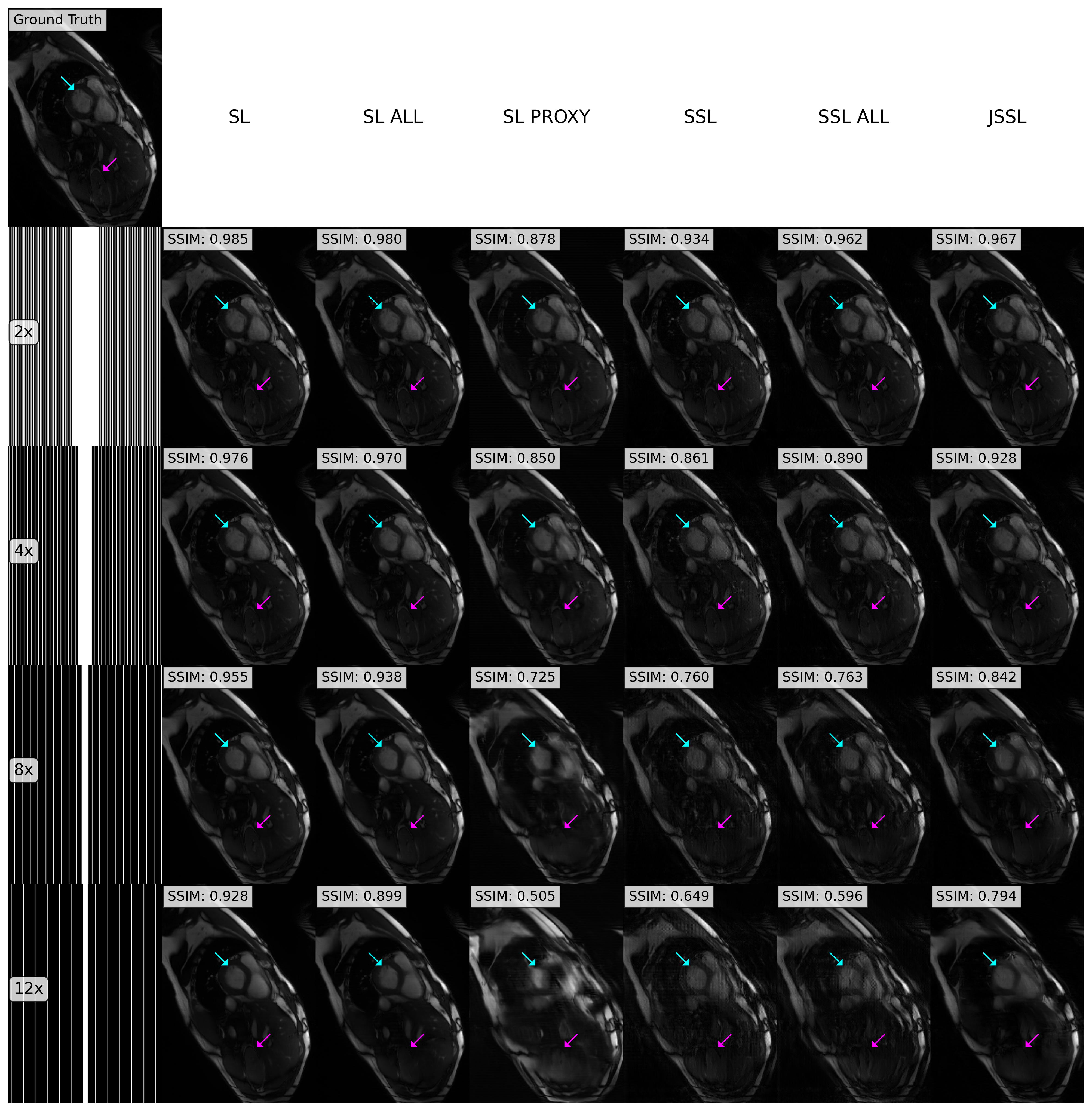}
    \caption{Example reconstructions of a slice from the cardiac cine dataset subsampled at different acceleration factors from the test set  (experiment set \textbf{B})  from each training setup.}
    \label{fig:example_figure_cardiac}
\end{figure*}

The results of our comparative studies, for experiment set \textbf{A} (prostate data as target, brain and knee data as proxies) and experiment set \textbf{B} (cardiac data as target, brain, knee, cardiac data as proxies) are visually represented with box plots in Figures \ref{fig:metrics} and \ref{fig:metrics-cine}. The corresponding metric averages and statistical significance results are presented in Tables \ref{tab:metrics} and \ref{tab:metrics-cine}.

As these results indicate, supervised methods consistently yielded the best reconstruction outcomes in both experiment sets, as anticipated. However, it is important to reiterate that our aim is not to compete with supervised learning but to demonstrate that in scenarios where supervision is not feasible, JSSL can enhance the performance of SSL.

From \Figure{metrics} and \Table{metrics}, we observe that in experiment set \textbf{A} with prostate data as the target, the JSSL setup achieved superior reconstruction results across all acceleration factors and metrics compared to both SSL and SSL utilizing both proxy and target datasets (SSL ALL). Moreover, the JSSL approach proved to be a strong competitor to supervised tasks (SL and SL ALL), particularly at acceleration factors of 2, 4, and 8. Conversely, using proxy datasets in SSL settings (SSL ALL) did not enhance reconstruction performance compared to SSL. Additionally, training on all data in a supervised manner (SL ALL) did not result in better quality metrics than SL alone.

When training only on proxy datasets (SL PROXY) in experiment set \textbf{A}, we found that performing out-of-distribution inference on the prostate dataset outperformed SSL in terms of reconstruction quality. However, JSSL still surpassed SL PROXY in terms of SSIM across all acceleration factors, while it showed similar results in terms of pSNR and NMSE (let alone at $R=2$, where SL PROXY performed better).

In experiment set \textbf{B} with cardiac data as the target, similar trends were observed. JSSL was the best method among SSL-based approaches, except at $R=12$, where the SSL method showed (non-statistically significant) better pSNR and NMSE values. Unlike \textbf{A}, training on all data in a self-supervised manner (SSL ALL) improved performance over SSL only in the target domain, as seen from \Figure{metrics-cine} and \Table{metrics-cine}. In this experiment, SL PROXY performed worse than all other methods, highlighting the challenges of generalizing from proxy domains to the cardiac domain.

For visual assessment, Figures \ref{fig:example_figure} and \ref{fig:example_figure_cardiac} and Supporting Information Figures \ref{fig:S4} and \ref{fig:S5}, depict sample image reconstructions across all acceleration factors and training setups. For lower acceleration factors (2 and 4), all methods accurately reconstructed the target prostate accelerated data. However, at higher accelerations, only the supervised and JSSL setups were able to reconstruct the images with fewer artifacts compared to the SSL and SSL ALL methods.  The same trend was observed for the cardiac dataset, where the SSL method produced visually weaker reconstructions, especially at acceleration factors of 4$\times$ to 12$\times$. Furthermore, for acceleration factors of 8 and 12, both SSL and SSL ALL experiments produced highly aliased images.

\subsubsection{Robustness to Model Choice Experiments Results}

\begin{table*}[!htb]
\centering
\caption{Robustness to model choice experiments results. 
}
\label{tab:comparison-unet-varnet}
\setlength{\tabcolsep}{0.8pt}
{\renewcommand{\arraystretch}{1.8}
\resizebox{1\textwidth}{!}{%
\begin{tabular}{ccccccccccc}
\hline
\multirow{2}{*}{\textbf{Architecture}} &
  \multirow{2}{*}{\textbf{Setup}} &
  \multicolumn{3}{c}{\textbf{4x}} &
  \multicolumn{3}{c}{\textbf{8x}} &
  \multicolumn{3}{c}{\textbf{16x}} \\ \cline{3-11} 
 &
   &
  SSIM &
  pSNR &
  NMSE &
  SSIM &
  pSNR &
  NMSE &
  SSIM &
  pSNR &
  NMSE \\ \hline
\multirow{2}{*}{U-Net} &
  SSL &
  0.854${\scriptscriptstyle\pm 0.031}$ &
  33.0${\scriptscriptstyle\pm 1.6}$ &
  0.013${\scriptscriptstyle\pm 0.004}$ &
  0.742${\scriptscriptstyle\pm 0.040}$ &
  29.4${\scriptscriptstyle\pm 1.4}$ &
  0.030${\scriptscriptstyle\pm 0.006}$ &
  0.651${\scriptscriptstyle\pm 0.051}$ &
  \textbf{26.7${\scriptscriptstyle\pm 1.5}$} &
  \textbf{0.055${\scriptscriptstyle\pm 0.009}$} \\
 &
  JSSL &
  \textbf{0.863${\scriptscriptstyle\pm 0.031}$} &
  \textbf{33.5${\scriptscriptstyle\pm 1.5}$} &
  \textbf{0.012${\scriptscriptstyle\pm 0.002}$} &
  \textbf{0.759${\scriptscriptstyle\pm 0.042}$} &
  \textbf{29.7${\scriptscriptstyle\pm 1.4}$} &
  \textbf{0.027${\scriptscriptstyle\pm 0.005}$} &
  \textbf{0.663${\scriptscriptstyle\pm 0.051}$} &
  26.7${\scriptscriptstyle\pm 1.4}$$^*$ &
  0.055${\scriptscriptstyle\pm 0.009}$$^*$ \\ \hline
\multirow{2}{*}{E2EVarNet} &
  SSL &
  0.874${\scriptscriptstyle\pm 0.029}$ &
  33.7${\scriptscriptstyle\pm 1.7}$ &
  0.011${\scriptscriptstyle\pm 0.003}$ &
  0.770${\scriptscriptstyle\pm 0.039}$ &
  30.0${\scriptscriptstyle\pm 1.4}$ &
  0.025${\scriptscriptstyle\pm 0.006}$ &
  0.670${\scriptscriptstyle\pm 0.051}$ &
  27.0${\scriptscriptstyle\pm 1.5}$ &
  0.051${\scriptscriptstyle\pm 0.009}$ \\
 &
  JSSL &
  \textbf{0.888${\scriptscriptstyle\pm 0.032}$} &
  \textbf{34.9${\scriptscriptstyle\pm 1.6}$} &
  \textbf{0.008${\scriptscriptstyle\pm 0.002}$} &
  \textbf{0.784${\scriptscriptstyle\pm 0.042}$} &
  \textbf{30.5${\scriptscriptstyle\pm 1.4}$} &
  \textbf{0.023${\scriptscriptstyle\pm 0.005}$} &
  \textbf{0.678${\scriptscriptstyle\pm 0.053}$} &
  \textbf{27.1${\scriptscriptstyle\pm 1.5}$} &
  \textbf{0.050${\scriptscriptstyle\pm 0.009}$} \\ \hline
\end{tabular}%
}
}
\end{table*}

The average results of our supplementary comparative studies to assess JSSL's robustness to different architecture choices are depicted via box plots in Supporting Information Figures \ref{fig:S2} and \ref{fig:S3} for U-Net and E2EVarNet, respectively. Corresponding average metrics are provided in \Table{comparison-unet-varnet}. From these results we observe alignment with our original findings: JSSL-trained models consistently outperform SSL-trained models for both architecture choices.

Furthermore, the superior performance of vSHARP and E2EVarNet compared to the U-Net model in both SSL and JSSL settings across all acceleration factors highlights the advantage of adopting physics-guided unrolled models for reconstruction. It is also worth mentioning that vSHARP consistently outperformed E2EVarNet at all accelerations.

\subsection{Alternative Configurations}
\label{sec:subsec4.2}

\begin{table*}[!hbt]
\centering
\caption{Average results for alternative configurations with the fastMRI prostate as target dataset. 
}
\resizebox{1\textwidth}{!}{%
\setlength{\tabcolsep}{0.8pt}
{\renewcommand{\arraystretch}{2}
\begin{tabular}{ccccccccccccc}
\hline
\multirow{2}{*}{\textbf{Setup}} &
  \multicolumn{3}{c}{\textbf{2x}} &
  \multicolumn{3}{c}{\textbf{4x}} &
  \multicolumn{3}{c}{\textbf{8x}} &
  \multicolumn{3}{c}{\textbf{16x}} \\ \cline{2-13} 
 &
  SSIM &
  pSNR &
  NMSE &
  SSIM &
  pSNR &
  NMSE &
  SSIM &
  pSNR &
  NMSE &
  SSIM &
  pSNR &
  NMSE \\ \hline
SSL Original &
  0.956${\scriptscriptstyle\pm 0.015}$ &
  38.8${\scriptscriptstyle\pm 2.6}$ &
  0.004${\scriptscriptstyle\pm 0.002}$ &
  0.891${\scriptscriptstyle\pm 0.030}$ &
  34.7${\scriptscriptstyle\pm 2.0}$ &
  0.009${\scriptscriptstyle\pm 0.003}$ &
  0.801${\scriptscriptstyle\pm 0.038}$ &
  31.1${\scriptscriptstyle\pm 1.5}$ &
  0.020${\scriptscriptstyle\pm 0.005}$ &
  0.707${\scriptscriptstyle\pm 0.050}$ &
  28.0${\scriptscriptstyle\pm 1.6}$ &
  0.041${\scriptscriptstyle\pm 0.008}$ \\
SSL ALL Original &
  0.953${\scriptscriptstyle\pm 0.016}$ &
  38.6${\scriptscriptstyle\pm 2.5}$ &
  0.004${\scriptscriptstyle\pm 0.002}$ &
  0.892${\scriptscriptstyle\pm 0.031}$ &
  34.8${\scriptscriptstyle\pm 2.0}$ &
  0.009${\scriptscriptstyle\pm 0.004}$ &
  0.801${\scriptscriptstyle\pm 0.041}$ &
  31.1${\scriptscriptstyle\pm 1.6}$ &
  0.020${\scriptscriptstyle\pm 0.006}$ &
  0.699${\scriptscriptstyle\pm 0.052}$ &
  27.8${\scriptscriptstyle\pm 1.6}$ &
  0.043${\scriptscriptstyle\pm 0.010}$ \\
JSSL Original &
  \textbf{0.965${\scriptscriptstyle\pm 0.015}$} &
  \textbf{39.5${\scriptscriptstyle\pm 2.8}$} &
  \textbf{0.003${\scriptscriptstyle\pm 0.002}$} &
  \textbf{0.918${\scriptscriptstyle\pm 0.026}$} &
  \textbf{36.4${\scriptscriptstyle\pm 1.9}$} &
  \textbf{0.006${\scriptscriptstyle\pm 0.002}$} &
  \textbf{0.842${\scriptscriptstyle\pm 0.038}$} &
  \textbf{32.5${\scriptscriptstyle\pm 1.6}$} &
  \textbf{0.015${\scriptscriptstyle\pm 0.004}$} &
  \textbf{0.752${\scriptscriptstyle\pm 0.053}$} &
  \textbf{29.3${\scriptscriptstyle\pm 1.6}$} &
  \textbf{0.030${\scriptscriptstyle\pm 0.007}$} \\ \hline
SSL ALL Oversamp. &
  \multicolumn{1}{r}{0.955${\scriptscriptstyle\pm 0.015}$} &
  \multicolumn{1}{r}{38.5${\scriptscriptstyle\pm 2.5}$} &
  \multicolumn{1}{r}{0.004${\scriptscriptstyle\pm 0.002}$} &
  \multicolumn{1}{r}{0.891${\scriptscriptstyle\pm 0.030}$} &
  \multicolumn{1}{r}{34.6${\scriptscriptstyle\pm 2.0}$} &
  \multicolumn{1}{r}{0.009${\scriptscriptstyle\pm 0.004}$} &
  \multicolumn{1}{r}{0.807${\scriptscriptstyle\pm 0.039}$} &
  \multicolumn{1}{r}{31.3${\scriptscriptstyle\pm 1.6}$} &
  \multicolumn{1}{r}{0.019${\scriptscriptstyle\pm 0.006}$} &
  \multicolumn{1}{r}{0.712${\scriptscriptstyle\pm 0.051}$} &
  \multicolumn{1}{r}{28.1${\scriptscriptstyle\pm 1.6}$} &
  \multicolumn{1}{r}{0.040${\scriptscriptstyle\pm 0.009}$} \\
JSSL Oversamp. &
  \textbf{0.968${\scriptscriptstyle\pm 0.012}$} &
  \textbf{41.0${\scriptscriptstyle\pm 2.1}$} &
  \textbf{0.002${\scriptscriptstyle\pm 0.001}$} &
  \textbf{0.919${\scriptscriptstyle\pm 0.026}$} &
  \textbf{36.7${\scriptscriptstyle\pm 1.9}$} &
  \textbf{0.006${\scriptscriptstyle\pm 0.002}$} &
  \textbf{0.842${\scriptscriptstyle\pm 0.038}$} &
  \textbf{32.6${\scriptscriptstyle\pm 1.6}$} &
  \textbf{0.014${\scriptscriptstyle\pm 0.004}$} &
  \textbf{0.749${\scriptscriptstyle\pm 0.052}$} &
  \textbf{29.2${\scriptscriptstyle\pm 1.6}$} &
  \textbf{0.031${\scriptscriptstyle\pm 0.008}$} \\ \hline
SSL (q=0.5) &
  0.957${\scriptscriptstyle\pm 0.014}$ &
  39.1${\scriptscriptstyle\pm 2.4}$ &
  0.003${\scriptscriptstyle\pm 0.002}$ &
  0.895${\scriptscriptstyle\pm 0.027}$ &
  35.0${\scriptscriptstyle\pm 1.9}$ &
  0.008${\scriptscriptstyle\pm 0.003}$ &
  0.817${\scriptscriptstyle\pm 0.038}$ &
  31.7${\scriptscriptstyle\pm 1.6}$ &
  0.017${\scriptscriptstyle\pm 0.005}$ &
  0.733${\scriptscriptstyle\pm 0.050}$ &
  28.9${\scriptscriptstyle\pm 1.6}$ &
  0.033${\scriptscriptstyle\pm 0.008}$ \\
SSL ALL (q=0.5) &
  0.956${\scriptscriptstyle\pm 0.015}$ &
  39.0${\scriptscriptstyle\pm 2.5}$ &
  0.004${\scriptscriptstyle\pm 0.002}$ &
  0.897${\scriptscriptstyle\pm 0.029}$ &
  35.2${\scriptscriptstyle\pm 2.0}$ &
  0.008${\scriptscriptstyle\pm 0.003}$ &
  0.816${\scriptscriptstyle\pm 0.039}$ &
  31.7${\scriptscriptstyle\pm 1.6}$ &
  0.018${\scriptscriptstyle\pm 0.005}$ &
  0.715${\scriptscriptstyle\pm 0.052}$ &
  28.3${\scriptscriptstyle\pm 1.6}$ &
  0.038${\scriptscriptstyle\pm 0.009}$ \\
JSSL (q=0.5) &
  \textbf{0.965${\scriptscriptstyle\pm 0.015}$} &
  \textbf{39.7${\scriptscriptstyle\pm 2.8}$} &
  \textbf{0.003${\scriptscriptstyle\pm 0.002}$} &
  \multicolumn{1}{r}{\textbf{0.919${\scriptscriptstyle\pm 0.026}$}} &
  \multicolumn{1}{r}{\textbf{36.6${\scriptscriptstyle\pm 1.8}$}} &
  \multicolumn{1}{r}{\textbf{0.006${\scriptscriptstyle\pm 0.002}$}} &
  \multicolumn{1}{r}{\textbf{0.842${\scriptscriptstyle\pm 0.037}$}} &
  \multicolumn{1}{r}{\textbf{32.6${\scriptscriptstyle\pm 1.6}$}} &
  \multicolumn{1}{r}{\textbf{0.014${\scriptscriptstyle\pm 0.004}$}} &
  \multicolumn{1}{r}{\textbf{0.742${\scriptscriptstyle\pm 0.052}$}} &
  \multicolumn{1}{r}{\textbf{29.0${\scriptscriptstyle\pm 1.6}$}} &
  \multicolumn{1}{r}{\textbf{0.032${\scriptscriptstyle\pm 0.007}$}} \\ \hline
SSL (w=10) &
  0.954${\scriptscriptstyle\pm 0.015}$ &
  38.4${\scriptscriptstyle\pm 2.5}$ &
  0.004${\scriptscriptstyle\pm 0.002}$ &
  0.893${\scriptscriptstyle\pm 0.029}$ &
  34.7${\scriptscriptstyle\pm 2.0}$ &
  0.009${\scriptscriptstyle\pm 0.004}$ &
  0.815${\scriptscriptstyle\pm 0.038}$ &
  31.6${\scriptscriptstyle\pm 1.6}$ &
  0.018${\scriptscriptstyle\pm 0.006}$ &
  0.726${\scriptscriptstyle\pm 0.048}$ &
  28.5${\scriptscriptstyle\pm 1.6}$ &
  0.036${\scriptscriptstyle\pm 0.008}$ \\
SSL ALL (w=10) &
  0.954${\scriptscriptstyle\pm 0.015}$ &
  38.4${\scriptscriptstyle\pm 2.5}$ &
  0.004${\scriptscriptstyle\pm 0.002}$ &
  0.890${\scriptscriptstyle\pm 0.029}$ &
  34.6${\scriptscriptstyle\pm 2.0}$ &
  0.009${\scriptscriptstyle\pm 0.004}$ &
  0.805${\scriptscriptstyle\pm 0.039}$ &
  31.2${\scriptscriptstyle\pm 1.6}$ &
  0.020${\scriptscriptstyle\pm 0.006}$ &
  0.710${\scriptscriptstyle\pm 0.052}$ &
  28.1${\scriptscriptstyle\pm 1.6}$ &
  0.040${\scriptscriptstyle\pm 0.009}$ \\
JSSL (w=10) &
  \textbf{0.958${\scriptscriptstyle\pm 0.016}$} &
  \textbf{38.7${\scriptscriptstyle\pm 2.6}$} &
  \textbf{0.004${\scriptscriptstyle\pm 0.002}$} &
  \textbf{0.916${\scriptscriptstyle\pm 0.026}$} &
  \textbf{36.4${\scriptscriptstyle\pm 1.8}$} &
  \textbf{0.006${\scriptscriptstyle\pm 0.002}$} &
  \textbf{0.839${\scriptscriptstyle\pm 0.038}$} &
  \textbf{32.5${\scriptscriptstyle\pm 1.6}$} &
  \textbf{0.015${\scriptscriptstyle\pm 0.004}$} &
  \textbf{0.748${\scriptscriptstyle\pm 0.052}$} &
  \textbf{29.2${\scriptscriptstyle\pm 1.6}$} &
  \textbf{0.031${\scriptscriptstyle\pm 0.008}$} \\ \hline
\end{tabular}%
}
}
\label{tab:ablation}
\end{table*}

Summarized in \Table{ablation}, we calculated the average evaluation metrics on the test set for our alternative configurations experiments, providing additional context to the JSSL approach.  These experiments consistently showcased the superior performance of JSSL over SSL setups, in line with our prior observations.  Interestingly, variations in the training hyper-parameters for JSSL, such as oversampling, partitioning ratio ($q=0.5$), and ACS window size ($w=10$), did not yield significant improvements or deteriorations in performance, except for an observable improvement in average pSNR at $R=2$.

Regarding the SSL setups, an observable enhancement was witnessed for 8$\times$ and 16$\times$ accelerated data when adopting a fixed partitioning ratio $q=0.5$ or a larger ACS window of $w^2=10^2$ pixels. However, this improvement was particularly evident in the SSL setup using solely the (subsampled) proxy dataset. Furthermore, the inclusion of proxy datasets within SSL configurations (SSL ALL) did not yield improvements in reconstruction performance, consistent with our earlier findings in the comparative study.

For further assessment, we provide in Supporting Information Appendix \ref{sec:ap4} box plots illustrating comprehensively the performance metrics, as well as sample reconstructions for each setup considered in the alternative configurations study.

\section{Discussion and Conclusion}
\label{sec:sec5}
Our study introduces Joint Supervised and Self-supervised Learning, a novel training framework designed to enhance the quality of MRI reconstructions when fully-sampled $k$-space data is unavailable for the target domain. JSSL leverages the strengths of supervised learning by incorporating fully-sampled proxy datasets alongside subsampled target datasets in a self-supervised manner. Evidently, this approach significantly outperforms traditional self-supervised learning methods, offering a promising alternative for MRI reconstruction in clinically challenging scenarios where acquiring fully-sampled data is not feasible.

Through comprehensive experiments, we demonstrated that JSSL consistently yields higher reconstruction quality across various acceleration factors compared to self-supervised learning alone. Notably, JSSL showed robust performance improvements even when the proxy datasets were from different anatomical regions, such as brain and knee MRI, or brain, knee and prostate MRI, compared to the target dataset, prostate or cardiac cine, respectively. This potentially indicates that the model effectively learns the underlying physics and reconstruction principles from the supervised task, which it can then apply to different target domains.

Our alternative configurations studies further confirmed that JSSL maintains its superiority under different training configurations, underscoring its robustness and adaptability to various clinical settings. In addition, we tested JSSL with different deep learning architectures, including the traditional U-Net and the state-of-the-art End-to-End Variational Network (E2EVarNet). These experiments revealed that JSSL's performance improvements are consistent regardless of the underlying model architecture, demonstrating the framework's robustness and independence to model choice.

While our experiments indicate that JSSL demonstrates improvements over conventional SSL methods, several limitations warrant discussion. Firstly, the efficacy of JSSL is highly dependent on the availability and quality of proxy datasets. Although datasets such as the fastMRI datasets contain fully-sampled data and are readily available, there might be instances where such datasets cannot be used. This could occur in cases where the anatomical regions of interest in the proxy datasets are not sufficiently similar to those in the target dataset, or where differences in imaging protocols and acquisition parameters introduce significant discrepancies.

For instance, in experiment set \textbf{A}, where the fastMRI prostate data served as the target domain and brain and knee fastMRI datasets were used as proxies, the SL PROXY setup showed relatively good performance, indicating that training with similar proxy domains can still be beneficial for out-of-distribution inference. However, in experiment set \textbf{B}, where the CMRxRecon cardiac data was the target and brain, knee, and prostate fastMRI datasets served as proxies, the performance of SL PROXY was significantly lower than all methods, highlighting that when proxies are dissimilar to the target, SL PROXY struggles to generalize effectively. In both scenarios, JSSL consistently surpassed SL PROXY, indicating that the combined supervised and self-supervised approach is more robust, regardless of the proxy dataset's similarity.

Additionally, the inclusion of proxy datasets in training can introduce biases, particularly if there are substantial differences between the proxy and target domains. This bias could potentially degrade the model's performance on the target dataset, as observed in some of our supervised learning experiments.

Moreover, similar to any DL-based method, JSSL's performance is influenced by the choice of loss functions for each component of the JSSL loss and their weighting in the loss $\mathcal{L}_{\boldsymbol{\psi}}^{\text{JSSL}}$. In our experiments, we employed identical dual-domain loss functions for each component and equal weighting for the SL and SSL components (see \Section{para3.7.4.3}). However, different loss and weighting choices might affect JSSL's performance.

JSSL performance also depends on the partitioning strategy used for subsampled data in self-supervised learning. While we adopted a Gaussian partitioning scheme, alternative strategies might yield different results and require further exploration. The optimal partitioning scheme may vary depending on the specific characteristics of the target and proxy datasets, as well as the desired reconstruction quality.

Lastly, our experiments are limited to comparing only one SSL method (SSDU) and does not consider other proposed self-supervised methodologies. However, the reason for comparing to SSDU only is that we consider it representative, as most SSL-based methods are derivatives of SSDU and still employ SSL-based losses to train their models (refer to \Section{sec2}). In addition, comparing to methods that train more than one model as their SSL task is outside the scope of this research, as this can introduce additional computational difficulties and are derivatives of the SSDU method. Our purpose is to compare JSSL and SSL training methods in their general forms.

Despite the advantages of JSSL, it is important to note that supervised learning remains the best option when fully-sampled ground truth data are available for the target dataset. Supervised methods provide the highest reconstruction quality due to the availability of accurate and complete training data.

Based on our empirical findings, we propose practical training ``rule-of-thumb'' guidelines when determining the approach for deep MRI reconstruction algorithms:

\begin{enumerate}[label=(\arabic*),leftmargin=*]
    \item If ground truth data are available for the target dataset, opt for supervised training.
    \item If ground truth data are not available for the target dataset but subsampled data are present, and ground truth data exist from other datasets (e.g. fastMRI or CMRxRecon datasets), consider adopting the JSSL approach.
    \item If only ground truth data from proxy datasets are available, training solely in the proxy domains with supervised learning can be beneficial, especially when the proxy domains are anatomically or contextually similar to the target domain. However, if the proxies are dissimilar, avoid this avenue.
    \item In cases where only subsampled data are accessible for the target dataset without ground truth data from other proxy datasets, proceed with self-supervision. If subsampled proxy data are available, incorporating these data might help. A fixed partitioning ratio might be preferable for high accelerations.
\end{enumerate}

In conclusion, JSSL presents a robust solution for MRI reconstruction in scenarios where fully-sampled $k$-space data is not available. By jointly leveraging supervised and self-supervised learning, JSSL significantly enhances reconstruction quality, especially at high acceleration factors. The proposed approach sets a new benchmark for self-supervised MRI reconstruction methods and opens new avenues for research and clinical applications.



\section*{Funding Information}
This work was supported by institutional grants of the Dutch Cancer Society and of the Dutch Ministry of Health, Welfare and Sport.

\section*{Acknowledgments}
The authors would like to acknowledge the Research High Performance Computing (RHPC) facility of the Netherlands Cancer Institute (NKI).

\clearpage

\Large
\noindent
\textbf{Joint Supervised and Self-Supervised Learning for MRI Reconstruction - Supporting Information}

\normalsize
\appendix

\setcounter{figure}{0} 
\setcounter{table}{0} 
\renewcommand{\thefigure}{S\arabic{figure}}
\renewcommand{\thetable}{S\arabic{table}}

\input{appendix1}
\clearpage
\input{appendix2}
\input{appendix3}

\input{appendix4}

\bibliographystyle{unsrt}  
\bibliography{bibliography}

\end{document}

%% file: extras.tex
\newcommand{\Figure}[1]{Figure~\ref{fig:#1}}

\newcommand{\Table}[1]{Table~\ref{tab:#1}}

\newcommand{\Equation}[1]{Equation~\ref{eq:#1}}

\newcommand{\Section}[1]{Section~\ref{sec:#1}}

\usepackage{amsmath}
\usepackage{bbold}
\renewcommand\vec{\mathbf}
\newcommand{\mat}{\mathbf}

\usepackage{enumitem}
\setlist[itemize]{noitemsep, topsep=0pt}
\setlist[enumerate]{nosep}

\usepackage{caption}
\usepackage{subcaption}

\usepackage{multirow}
\usepackage{graphicx}

\usepackage{sidecap}

\usepackage{float}

\usepackage[flushleft]{threeparttable}

\usepackage{amssymb}
\usepackage{bm}

\newcommand{\nosemic}{\renewcommand{\@endalgocfline}{\relax}}
\newcommand{\dosemic}{\renewcommand{\@endalgocfline}{\algocf@endline}}
\let\oldnl\nl
\newcommand{\nonl}{\renewcommand{\nl}{\let\nl\oldnl}}
\usepackage{afterpage} 

\usepackage{amsthm}
\newtheorem{prop}{Proposition}

\usepackage{lipsum}  

\usepackage{wrapfig}

\usepackage{lmodern}

\usepackage{afterpage}

\usepackage{hyperref}

\usepackage[linesnumbered,ruled,vlined]{algorithm2e}
\usepackage {algpseudocode}
\usepackage{algorithmicx}
\usepackage{algcompatible}

%% file: appendix1.tex
\section{JSSL Theoretical Motivation}
\label{sec:ap1}
\noindent
In this appendix we provide theoretical motivations along with proofs for the JSSL method  presented in \Section{subsec3.8} of the main paper.

\begin{prop}
    Consider two distributions $p_i, \, i=1,2$ with means and variances ${\mu}_i, {\sigma}_i, \, i=1,2$, with unknown ${\mu}_1$, and ${\mu}_1 \neq {\mu}_2$. Then if $(\mu_1 - \mu_2)^2 < c \frac{{\sigma}_1^2} N$ for some $c \in (0, 1)$ and $N \in \mathbb{Z}^{+}$, then $\tilde{{x}} = \frac{1}{N+K} \sum_{i=1}^{N+K} {x}_i$ is a lower-variance estimator of $\mu_1$ compared to $\overline{{x}} = \frac{1}{N} \sum_{i=1}^N {x}_i$, where $\left\{{x}^{(i)} \sim p_1\right\}_{i=1}^{N}$ and $\left\{{x}^{(N+i)} \sim p_2 \right\}_{i=1}^{K}$ for a choice of a large $K \in \mathbb{Z}^{+} $.
\end{prop}

\begin{proof}
    Assume a mixture distribution:
\begin{equation*}
    p_{\pi}(x) = \pi \mathcal{N}(x | {\mu}_1, {\sigma}_1^2) + (1-\pi) \mathcal{N}(x | {\mu}_2, {\sigma}_2^2).
\end{equation*}
\noindent
It is then straightforward to compute:
\begin{equation*}
    \begin{gathered}
           \mathbb E\left[p_{\pi}\right] = \pi {\mu}_1 + (1- \pi) {\mu}_2
    \end{gathered}
\end{equation*}
and,
\begin{equation*}
    \mathbb{V} \left[p_{\pi}\right] = \pi {\sigma}_1^2 + (1-\pi) {\sigma}_2^2 + \pi (1-\pi) ({\mu}_2 - {\mu}_1)^2.
\end{equation*}

\noindent
Drawing $\displaystyle \left\{{x}^{(i)} \sim p_1\right\}_{i=1}^{N}$ and $\displaystyle 
\left\{{x}^{(N+i)} \sim p_2 \right\}_{i=1}^{K}$, is approximately equivalent to drawing $N+K$ samples from the mixture $p_{\pi}$ with $\pi = \frac{N}{N+K}$. Using bias-variance decomposition, we can compute the expected mean squared errors for the two estimators:

\begin{equation*}
   \mathbb E\left[(\overline{{x}} - {\mu}_1)^2 \right] = \frac{{\sigma}_1^2}{N},
\label{eq.mse_est2_2}
\end{equation*}
and,
\begin{equation*}
    \mathbb{E}\left[(\tilde{{x}} - {\mu}_1)^2 \right] = 
   (1-\pi)^2 (\mu_1 - \mu_2)^2 + \frac{\pi \sigma_1^2 + (1-\pi) \sigma_2^2 + \pi (1-\pi) (\mu_1 - \mu_2)^2}{N+K}.
\label{eq.mse_est2}
\end{equation*}

If $(\mu_1 - \mu_2)^2 < c \frac{{\sigma}_1^2} N$ for some $c \in (0, 1)$, then  taking the limit $K \to \infty$ and thus $\pi \to 0$, we observe that \begin{equation*}
\mathbb{E}\left[(\tilde{{x}} - {\mu}_1)^2  \right] \to (\mu_1 - \mu_2)^2 < c \frac{{\sigma}_1^2}{N} < \frac{{\sigma}_1^2}{N} = \mathbb E\left[(\overline{{x}} - {\mu}_1)^2 \right].
\end{equation*}

\end{proof}

\break
\begin{prop}
Let $\bm x \sim \mathcal N(\bm 0, \sigma^2 \bm I_{p})$ be $\mathbb R^p$-valued isotropic Gaussian random vector and $y, \tilde y$ be random variables with $p(y|\bm x) = \mathcal N(y| \bm w^T \bm x, \varepsilon^2)$ and $p(\tilde y|\bm x) = \mathcal N(\tilde y| \bm{\tilde w}^T \bm x, \tilde \varepsilon^2)$ for some $\bm w, \bm{\tilde w} \in \mathbb R^p$. Let $\mathcal T = \{(\bm x_1, \tilde y_1),\dots,  (\bm x_K, \tilde y_K)\}$ be a training data set with $K > p$ and consider a maximum likelihood estimator $\widehat y(\bm x; \mathcal T)$ for $y$ given $\bm x$, computed using $\mathcal T$. Then the following holds:
\begin{enumerate}
\item $\mathrm{Bias}_{\mathcal T}[\widehat y(\bm x; \mathcal T)] = (\bm{\tilde w}^T - \bm w^T) \bm x.$
\item $\mathrm{Var}_{\mathcal T}[\widehat y(\bm x; \mathcal T)] =  \frac{\tilde \varepsilon^2}{\sigma^2 K} \| \bm x \|_2^2$.
\item $\mathbb E_{(\bm x, y)} [\widehat y(\bm x; \mathcal T) - y]^2 \leq  p \sigma^2 \| \bm {\tilde w} - \bm w \|_2^2  + \frac{p \tilde \varepsilon^2}{K}  + \varepsilon^2$
\end{enumerate}
\end{prop}
\begin{proof}
Let $\bm{\tilde w}_{\mathrm{MLE}} = (\bm X^T \bm X)^{-1} \bm X^T \bm{\tilde y}$ be the MLE estimator for $\bm {\tilde w}$, where the $K$ rows of $\bm X \in \mathbb R^{K \times p}$ are given by $\bm x_1^T, \dots, \bm x_K^T$ and the vector $\bm{\tilde y}$ is defined as $\bm{\tilde y} := (\tilde y_1, \dots, \tilde y_K) \in \mathbb R^K$. Since $K > p$, matrix $\bm X$ has full column rank almost surely and thus $\bm X^T \bm X$ is almost surely invertible. Observe that 
\begin{equation*}
\mathbb E_{\mathcal T} [\bm{\tilde w}_{\mathrm{MLE}}^T] = \mathbb E_{\mathcal T} [(\bm{\tilde \varepsilon}^T + \bm{\tilde w}^T \bm X^T) \bm X (\bm X^T \bm X)^{-1}] =\bm{\tilde w}^T,
\end{equation*}
since $\bm{\tilde \varepsilon} := \bm{\tilde y} - \bm X \bm{\tilde w}$ has zero mean, is independent from $\bm x_i$'s and the expectation $\mathbb E_{\mathcal T}[\cdot]$ can be rewritten as $\mathbb E_{\bm x_1, \dots, \bm x_K} [\mathbb E_{\bm{\tilde \varepsilon}}[\cdot]]$.
By definition of estimator bias, 
\begin{equation*}
\mathrm{Bias}_{\mathcal T}[\widehat y(\bm x; \mathcal T)] = \mathbb E_{\mathcal T}[\widehat y(\bm x; \mathcal T)] - \mathbb E_{y | \bm x} y = \mathbb E_{\mathcal T} [\bm{\tilde w}_{\mathrm{MLE}}^T] \bm x - \bm w^T \bm x =(\bm{\tilde w}^T - \bm w^T) \bm x.
\end{equation*}
Next,
\begin{align*}
&\mathrm{Var}_{\mathcal T}[\widehat y(\bm x; \mathcal T)] = \mathbb E_{\mathcal T} [\mathbb E_{\mathcal T}[\widehat y(\bm x; \mathcal T)] - \widehat y(\bm x; \mathcal T)]^2 =\\
&= \mathbb E_{\mathcal T} [ \bm{\tilde w}^T \bm x - (\bm{\tilde \varepsilon}^T + \bm{\tilde w}^T \bm X^T) \bm X (\bm X^T \bm X)^{-1} \bm x]^2 = \mathbb E_{\mathcal T} [ \bm{\tilde \varepsilon}^T \bm X (\bm X^T \bm X)^{-1} \bm x]^2.
\end{align*}
The scalar $(\bm{\tilde \varepsilon}^T \bm X (\bm X^T \bm X)^{-1} \bm x)^2$ can be equivalently written as 
\begin{equation*}
(\bm{\tilde \varepsilon}^T \bm X (\bm X^T \bm X)^{-1} \bm x)^T (\bm{\tilde \varepsilon}^T \bm X (\bm X^T \bm X)^{-1} \bm x) =
\bm x^T (\bm X^T \bm X)^{-1} \bm X^T \bm{\tilde \varepsilon} \bm{\tilde \varepsilon}^T \bm X (\bm X^T \bm X)^{-1} \bm x.    
\end{equation*}
Using that $\mathbb E_{\mathcal T}[\cdot] = \mathbb E_{\bm x_1, \dots, \bm x_k} [\mathbb E_{\bm{\tilde \varepsilon}}[\cdot]]$, we deduce that
\begin{align*}
&\mathbb E_{\mathcal T} [ \bm{\tilde \varepsilon}^T \bm X (\bm X^T \bm X)^{-1} \bm x]^2 
= \mathbb E_{\bm x_1, \dots, \bm x_K} [ \bm x^T (\bm X^T \bm X)^{-1} \bm X^T \mathbb E_{\bm{\tilde \varepsilon}}[\bm{\tilde \varepsilon} \bm{\tilde \varepsilon}^T] \bm X (\bm X^T \bm X)^{-1} \bm x ]=\\
&=\tilde \varepsilon^2 \mathbb E_{\bm x_1, \dots, \bm x_K} [ \bm x^T (\bm X^T \bm X)^{-1} \bm x ] = \tilde \varepsilon^2 \mathbb E_{\bm x_1, \dots, \bm x_K} [\mathrm{tr}(\bm x^T (\bm X^T \bm X)^{-1} \bm x) ] =\\
&= \tilde \varepsilon^2 \mathbb E_{\bm x_1, \dots, \bm x_K} [\mathrm{tr}(\bm x \bm x^T (\bm X^T \bm X)^{-1})] = \tilde \varepsilon^2  \mathrm{tr}(\bm x \bm x^T \mathbb E_{\bm x_1, \dots, \bm x_K} [(\bm X^T \bm X)^{-1}]),
\end{align*}
where we use cyclic property of the trace and the fact that $z = \mathrm{tr}(z)$ for a scalar $z$. To compute $\mathbb E_{\bm x_1, \dots, \bm x_K} [(\bm X^T \bm X)^{-1}]$, we note that, by definition, $\bm X^T \bm X$ follows Wishart distribution $\mathcal W_p(\sigma^2 \bm I_p, K)$ with $K$ degrees of freedom and thus $(\bm X^T \bm X)^{-1}$ follows inverse Wishart distribution $\mathcal W^{-1}_p(\sigma^{-2} \bm I_p, K + p + 1)$, whose mean equals $\frac{\bm I_p}{\sigma^2K}$. Combining this with the previous results, we conclude
\begin{align*}
&\mathrm{Var}_{\mathcal T}[\widehat y(\bm x; \mathcal T)] = \frac{\tilde \varepsilon^2}{\sigma^2 K}  \mathrm{tr}(\bm x \bm x^T) = \frac{\tilde \varepsilon^2}{\sigma^2 K} \| \bm x \|_2^2.
\end{align*}
The final estimate follows from the first two identities and the bias-variance decomposition.
\end{proof}

%% file: appendix2.tex
\section{Experiments}
\label{sec:ap2}

\subsection{Reconstruction Network - vSHARP}
In our main experiments, we employed the variable Splitting Half-quadratic ADMM algorithm for Reconstruction of inverse-Problems (vSHARP) as our reconstruction network, which is an unrolled physics-guided DL-based method \cite{yiasemis2023vsharp} that has previously been applied in accelerated brain, prostate and dynamic cardiac MRI reconstruction  \cite{yiasemis2023vsharp, yiasemis2023deep}. The vSHARP algorithm incorporates the half-quadratic variable splitting method to the optimization problem presented in \eqref{eq:variational_problem}, introducing an auxiliary variable $\vec{z}$:
\begin{equation}
    \min_{\vec{x}^{'}, \vec{z}}\frac{1}{2}\left|\left| \mathcal{A}_{\vec{M}, \mat{S}}(\vec{x}^{'}) - \Tilde{\vec{y}}_{\mat{M}}\right|\right|_2^2 +  \mathcal{G}(\vec{z}) \quad \text{s.t. } \vec{x}^{'} = \vec{z}.
    \label{eq:vsharp_var}
\end{equation}

Subsequently, \eqref{eq:vsharp_var} is iteratively unrolled over $T$ iterations using the Alternating Direction Method of Multipliers (ADMM). The ADMM formulation consists of three key steps: (a) a denoising step to refine the auxiliary variable $\vec{z}$, (b) data consistency for the target image $\vec{x}$, and (c) an update for the Lagrange Multipliers $\vec{u}$ introduced by ADMM:
\begin{subequations}
    \begin{gather}
        \vec{z}_{t+1} =  \mathcal{H}_{\boldsymbol{\psi}_{t+1}}\left( \vec{z}_{t}, \vec{x}_{t}, \vec{u}_{t} /{\mu}_{t+1}\right),\label{eq:z_step}\\
        \vec{x}_{t+1}  =  \arg\min_{\vec{x}^{'}} \Big|\Big| \mathcal{A}_{\vec{M}, \mat{S}}  (\vec{x}^{'})  \, -  \,  \Tilde{\vec{y}}_{\mat{M}}\Big|\Big|_2^2 +  \mu\left|\left|\vec{x}^{'} - \vec{z}_{t+1} + \vec{u}_{t}/{\mu}_{t+1} \right|\right|_2^2,
        \label{eq:x_step}\\
        \vec{u}_{t+1} = \vec{u}_{t} + \mu_{t+1} \left(\vec{x}_{t+1} - \vec{z}_{t+1}\right).
        \label{eq:u_step}
    \end{gather}
\end{subequations}

In \eqref{eq:z_step}, $\mathcal{H}_{\boldsymbol{\psi}_{t+1}}$ denotes a convolutional based DL image denoiser with trainable parameters $\boldsymbol{\psi}_{t+1}$, and $\eta_{t+1}$ a trainable learning rate. At each iteration, $\mathcal{H}_{\boldsymbol{\psi}_{t+1}}$ takes as input the previous predictions of the three variables and outputs an estimation of the auxiliary variable $\vec{z}$. Equation \ref{eq:x_step} is solved numerically by unrolling further a gradient descent scheme over $T_{\vec{x}}$ iterations. The last step in \eqref{eq:u_step}, involves a straightforward computation.  The initial approximations for $\vec{x}$ and $\vec{z}$ are taken as: 	 $\vec{x}_{0},\, \vec{z}_0 = \mathcal{R}_{\mat{S}} \circ \mathcal{F}^{-1} \left(\Tilde{\vec{y}}_{\mat{M}}\right).$  Moreover, for $\vec{u}_{0}$, vSHARP employs a trainable replication-padding and dilated convolutional-based network represented by $\mathcal{U}_{\boldsymbol{\phi}}$: $\vec{u}_{0} = \mathcal{U}_{\boldsymbol{\phi}}(\vec{x}_0).$

\subsection{Datasets Information}

\input{tab1}

\subsection{SSL Subsampling Partitioning}

Let $\mat{M}_{i}$ denote the sampling set. Here we describe $\mat{M}_{i}$ as a sampling mask in the form of a squared array of size $n = n_x \times n_y$ such that:
\begin{equation*}
    (\mat{M}_{i})_{kj} =  \Bigg\{\begin{array}{ll}
        1, & \text{if } (k, j) \text{ is sampled}\\
        0, & \text{if } (k, j) \text{ is not sampled}.\\
        \end{array}
\end{equation*}

The set $\mat{\Theta}_{i}$ is obtained by selecting elements from $\mat{M}_{i}$ using a variable density 2D Gaussian scheme with a standard deviation of $\sigma$ pixels and mean vector as the center of the sampling set $\vec{M}_{i}$, up to the number of elements determined by a ratio $q_i$, determined such that $q_i = \frac{|\mat{\Theta}_{i}|}{|\mat{M}_{i}|}$, where $| \cdot |$ here denotes the cardinality. Mathematically, the selection process for $\mat{\Theta}_{i}$ from $\mat{M}_{i}$ can be described  by the following algorithm:

\begin{algorithm}[H]
    \SetAlgoLined
    \KwData{Squared array $\mat{M}_{i}$ of size $n_x \times n_y$, ratio $0<q_i<1$, standard deviation $\sigma$}
    \KwResult{Set $\mat{\Theta}_{i}$}
    Initialize $\mat{\Theta}_{i}$ as an array of zeros of the same size as $\mat{M}_{i}$\;
    \While{$\frac{|\mat{\Theta}_{i}|}{|\mat{M}_{i}|} < q_i$}{
        Generate $(k, j)$ from $\mathcal{N}\left([\frac{n_x}{2}, \frac{n_y}{2}], \, \sigma^2\mat{I}_2\right)$\;
        \If{$(\mat{\Theta}_{i})_{kj} == 0$}{
            $(\mat{\Theta}_{i})_{kj} \leftarrow 1$\;
        }
    }
    \caption{Generation of $\mat{\Theta}_{i}$ using Gaussian Sampling}
\end{algorithm}

Subsequently, to partition $\mat{M}_{i}$, we set ${\mat{\Lambda}_{i}} = {\mat{M}_{i} \smallsetminus \mat{\Theta}_{i}}$. Note that by selecting $q_i=0$ then $\mat{\Theta}_{i} = \emptyset$, and for $q_i=1$ then $\mat{\Theta}_{i} = \mat{M}_i$.

For our comparison study in \Section{subsec4.1} of the main paper for SSL and JSSL experiments we randomly selected the ratio $q_i$ between 0.3, 0.4, 0.5, 0.6, 0.7 and 0.8. For our alternative configurations study in Section 4.2, we employed an identical partitioning ratio selection except for the case of a fixed ratio of $q_i = 0.5$. In all our JSSL and SSL experiments we used $\sigma = 3.5$.

\input{fig1}

\subsection{Choice of Loss Functions}
Here we provide the mathematical definitions of the loss function components we employed in our experiments.

\begin{itemize}
  \item Image Domain Loss Functions
  \begin{itemize}
    \item Structural Similarity Index Measure (SSIM) Loss
    
        \begin{equation}
                \mathcal{L}_\text{SSIM} := 1 - \text{SSIM}, \quad  \text{SSIM}(\vec{a},\,\vec{b}) =
            \frac{1}{N}\sum_{i=1}^{N} \frac{(2\mu_{\vec{a}_i}\mu_{\vec{b}_i} + \gamma_1)(2\sigma_{\vec{a}_i\vec{b}_i} + \gamma_2)}{({\mu^2_{\vec{a}_i}} +{\mu^2_{\vec{b}_i}} + \gamma_1)({\sigma^2_{\vec{a}_i}} + {\sigma^2_{\vec{b}_i}} + \gamma_2)},
            \label{eq:ssim_metric} 
        \end{equation}
    
        where $\vec{a}_i, \vec{b}_i, i=1,...,N$ represent $7\times 7$ square windows of  $\vec{a}, \vec{b}$, respectively, and  $\gamma_1 = 0.01$, $\gamma_1 = 0.03$. Additionally, $\mu_{\vec{a}_i}$, $\mu_{\vec{b}_i}$ denote the means of each window, $\sigma_{\vec{a}_i}$ and $\sigma_{\vec{b}_i}$ represent the corresponding standard deviations. Lastly, $\sigma_{\vec{a}_i\vec{b}_i}$ represents the covariance between $\vec{a}_i$ and $\vec{b}_i$.
    
    \item High Frequency Error Norm (HFEN)

    \begin{equation}
        \mathcal{L}_{\text{HFEN}_k} := {\text{HFEN}_k}, \quad  {\text{HFEN}_k}(\vec{a},\,\vec{b})  = \, \frac{|| \mathcal{G}(\vec{a}) - \mathcal{G}(\vec{b}) ||_k}{||\mathcal{G}(\vec{b})||_k},
    \label{eq:hfen}
\end{equation}

 where $\mathcal{G}$ is a Laplacian-of-Gaussian filter  \cite{5617283} with kernel of size $15\times 15$ and with a standard deviation of 2.5, and $k\,=\, 1 \text{ or } 2$.

    \item Mean Average Error (MAE / $L_1$) Loss
    \begin{equation}
        \mathcal{L}_1(\vec{a},\,\vec{b}) = || \vec{a} - \vec{b} ||_1 = \sum_{i=1}^n |a_{i} - b_{i}|
    \end{equation}
    \end{itemize}
    
    \item $k$-space Domain Loss Functions

    \begin{itemize}
        \item Normalized Mean Squared Error (NMSE)
        \begin{equation}
            \mathcal{L}_{\text{NMSE}} := \text{NMSE}, \quad  \text{NMSE}(\vec{a},\, \vec{b})\,= \, \frac{||\vec{a}\,-\,\vec{b}||_2^2}{||\vec{a}||_2^2}\,= \, \frac{\sum_{i=1}^n(a_{i} - b_{i})^2}{\sum_{i=1}^n a_{i}^{2}}.
            \label{eq:nmse}
        \end{equation}
        \item Normalized Mean Average Error (NMAE)
        \begin{equation}
            \mathcal{L}_{\text{NMAE}} := \text{NMAE}, \quad  \text{NMAE}(\vec{a},\, \vec{b})\,= \, \frac{||\vec{a}\,-\,\vec{b}||_1}{||\vec{a}||_1}\,= \, \frac{\sum_{i=1}^n |a_{i} - {b}_{i}|}{\sum_{i=1}^n |a_{i}|}.
            \label{eq:nmae}
        \end{equation}
    \end{itemize}

\end{itemize}

\newpage

%% file: tab1.tex
\begin{table}[!ht]
\centering
\caption{Dataset characteristics and splits.}
\label{tab:S1}
\resizebox{\textwidth}{!}{%
{\renewcommand{\arraystretch}{2}
\begin{tabular}{|cc|c|c|c|c|}
\hline
\multicolumn{2}{|c|}{\textbf{Dataset}}                                                                                                        & \textbf{fastMRI Knee \cite{zbontar2019fastmri}}                                           & \textbf{fastMRI Brain \cite{zbontar2019fastmri}}                                   & \textbf{fastMRI Prostate \cite{cmrxrecon,cmrxrecondataset}}                           & \textbf{CMRxRecon Cine \cite{tibrewala2023fastmri}}                                \\ \hline
\multicolumn{2}{|c|}{\textbf{Field Strength}}                                                                                                 & 1.5 T, 3.0 T                                                                              & 1.5 T, 3.0 T                                                                       & 3.0 T                                                                                 & 3.0 T                                                                              \\ \hline
\multicolumn{2}{|c|}{\textbf{Sequence}}                                                                                                       & \begin{tabular}[c]{@{}c@{}}Proton Density with\\ and without fat suppression\end{tabular} & \begin{tabular}[c]{@{}c@{}}T1-w pre and post \\ contrast, T2-w, FLAIR\end{tabular} & T2-w                                                                                  & Cine                                                                               \\ \hline
\multicolumn{2}{|c|}{\textbf{Subjects}}                                                                                                       & \begin{tabular}[c]{@{}c@{}}Healthy or \\ Abnormality present\end{tabular}                 & \begin{tabular}[c]{@{}c@{}}Healthy or \\ Pathology present\end{tabular}            & Cancer Patients                                                                       & Healthy                                                                            \\ \hline
\multicolumn{2}{|c|}{\textbf{Acquisition}}                                                                                                    & Cartesian                                                                                 & Cartesian                                                                          & Cartesian                                                                             & Cartesian                                                                          \\ \hline
\multicolumn{2}{|c|}{\textbf{\begin{tabular}[c]{@{}c@{}}Fully Sampled or \\ Subsampled\end{tabular}}}                                         & Fully Sampled                                                                             & Fully Sampled                                                                      & \begin{tabular}[c]{@{}c@{}}Three averages (2x) / \\ GRAPPA reconstructed\end{tabular} & \begin{tabular}[c]{@{}c@{}}One average (3x) / \\ GRAPPA reconstructed\end{tabular} \\ \hline
\multicolumn{2}{|c|}{\textbf{No. Coils}}                                                                                                      & 15                                                                                        & 2-24                                                                               & 10-30                                                                                 & 10                                                                                 \\ \hline
\multicolumn{2}{|c|}{\textbf{No. Volumes Used}}                                                                                               & 973                                                                                       & 2,991                                                                              & 312                                                                                   & 473                                                                                \\ \hline
\multicolumn{2}{|c|}{\textbf{No. Slices Used}}                                                                                                & 34,742                                                                                    & 47,426                                                                             & 9,508                                                                                 & 3,185                                                                              \\ \hline
\multicolumn{1}{|c|}{\multirow{3}{*}{\textbf{\begin{tabular}[c]{@{}c@{}}Split Size\\ (No. Volumes/\\ No. Slices)\end{tabular}}}} & Training   & 973 / 34,742                                                                              & 2,991 / 47,426                                                                     & 218 / 6,647                                                                           & 203 / 1,364                                                                        \\ \cline{2-6} 
\multicolumn{1}{|c|}{}                                                                                                           & Validation & -                                                                                         & -                                                                                  & 48 / 1,462                                                                            & 111 / 731                                                                          \\ \cline{2-6} 
\multicolumn{1}{|c|}{}                                                                                                           & Test       & -                                                                                         & -                                                                                  & 46 / 1,399                                                                            & 159 / 1,090                                                                        \\ \hline
\end{tabular}
}
}
\end{table}

%% file: fig1.tex
\begin{figure}[!hbt]
    \centering
    \begin{subfigure}{0.3\textwidth}
        \centering
        \includegraphics[width=1\columnwidth]{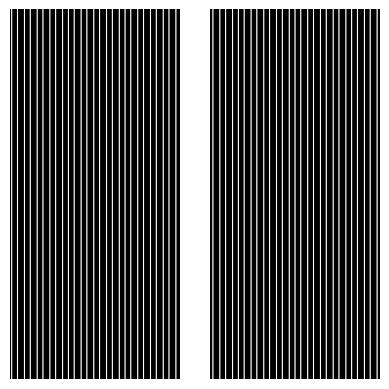}
        \caption{$\mat{M}_{i}$}
    \end{subfigure}%
        \hfill
    \begin{subfigure}{0.3\textwidth}
        \centering
        \includegraphics[width=1\columnwidth]{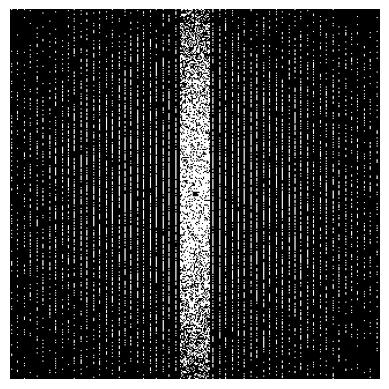}
        \caption{$\mat{\Theta}_{i}$}
    \end{subfigure}
    \hfill
    \begin{subfigure}{0.3\textwidth}
        \centering
        \includegraphics[width=1\columnwidth]{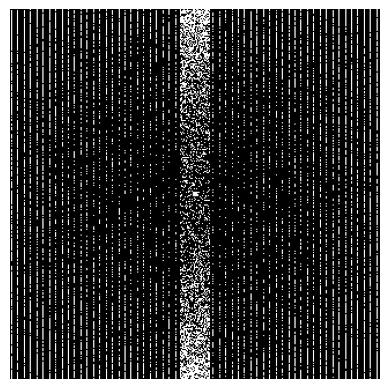}
        \caption{$\mat{\Lambda}_{i} = \mat{M}_{i}  \smallsetminus  \mat{\Theta}_{i}$}
    \end{subfigure}%
    \caption{Example of SSL subsampling partitioning with a ratio $q=\frac{1}{2}$ and $w=4$. 
    }
    \label{fig:S1}
\end{figure}

%% file: appendix3.tex
\section{Supplementary Experiments}
\label{sec:ap3}

In this section, we present supplementary experiments aimed at further validating the efficacy of our proposed JSSL method. These experiments involve a comparative analysis between JSSL and traditional SSL MRI Reconstruction. We adapt the methodologies outlined in \Section{sec3} of our main paper, utilizing two distinct reconstruction models instead of the vSHARP architecture:

\begin{itemize}
    \item Utilizing a plain image domain U-Net \cite{Ronneberger2015}, a non-physics-based model that takes an undersampled-reconstructed image as input and refines it. Specifically, we employ a U-Net with four scales and 64 filters in the first channel.
    
    \item Employing an End-to-end Variational Network (E2EVarNet) \cite{Sriram2020}, a physics-based model that executes a gradient descent-like optimization scheme in the $k$-space domain. For E2EVarNet, we perform 6 optimization steps using U-Nets with four scales and 16 filters in the first scale.
\end{itemize}

To estimate sensitivity maps for both architectures, an identical Sensitivity Map Estimation (SME) module was integrated, mirroring the experimental setup outlined in our primary paper.

Both models underwent training and evaluation on data subsampled with acceleration factors of 4, 8, and 16, with ACS ratios of 8\%, 4\%, and 2\% of the data shape.  Choices of hyperparameters for JSSL and SSL are the same as in the comparative experiments presented in \Section{sec4}. Additionally, choices for proxy and target datasets, as well as data splits, are also the same as in the main paper.

Experimental setups were executed on NVIDIA A100 80GB GPUs, utilizing 2 GPUs for U-Net and 1 GPU for E2EVarNet. We employed batch sizes of 2 and 4 for U-Net and E2EVarNet, respectively, on each GPU. The optimization procedures, initial learning rates, and the employed optimizers aligned with those utilized in the main paper. 

\input{tab2}

\Table{S2} details the model specifics for all considered architectures presented in both the main paper and this section.

\newpage

\input{fig2}
\input{fig3}

\clearpage 

%% file: tab2.tex
\begin{table}[!htb]
\centering
\caption{Model architectures parameters.}
\resizebox{\textwidth}{!}{%
{\renewcommand{\arraystretch}{2}
\begin{tabular}{ccccccccccc}
\cline{1-3} \cline{5-11}
\textbf{Model} &           & \textbf{\begin{tabular}[c]{@{}c@{}}Parameter \\ Count (millions)\end{tabular}} & \textbf{} & \textbf{Physics Model}               &  & \textbf{\begin{tabular}[c]{@{}c@{}}Training\\ Iterations (k)\end{tabular}} & \textbf{} & \textbf{\begin{tabular}[c]{@{}c@{}}Learning Rate\\ Reduction Schedule\end{tabular}} &  & \textbf{\begin{tabular}[c]{@{}c@{}}Inference Time (s)\\ per volume\end{tabular}} \\ \cline{1-1} \cline{3-3} \cline{5-5} \cline{7-7} \cline{9-9} \cline{11-11} 
vSHARP         &           & 95                                                                             &           & ADMM                                 &  & 700                                                                        &           & 150k                                                                                &  & 17.7                                                                             \\
U-Net          &           & 33                                                                             &           & -                                    &  & 375                                                                        &           & 75k                                                                                 &  & 13.1                                                                             \\
E2EVarNet      & \textbf{} & 13.5                                                                           &           & Gradient Descent in $k$-space Domain &  & 250                                                                        &           & 50k                                                                                 &  & 13.9                                                                             \\ \hline
\end{tabular}%
}
}
\label{tab:S2}
\end{table}

%% file: fig2.tex
\begin{figure}[!ht]
    \centering
    \includegraphics[width=\textwidth]{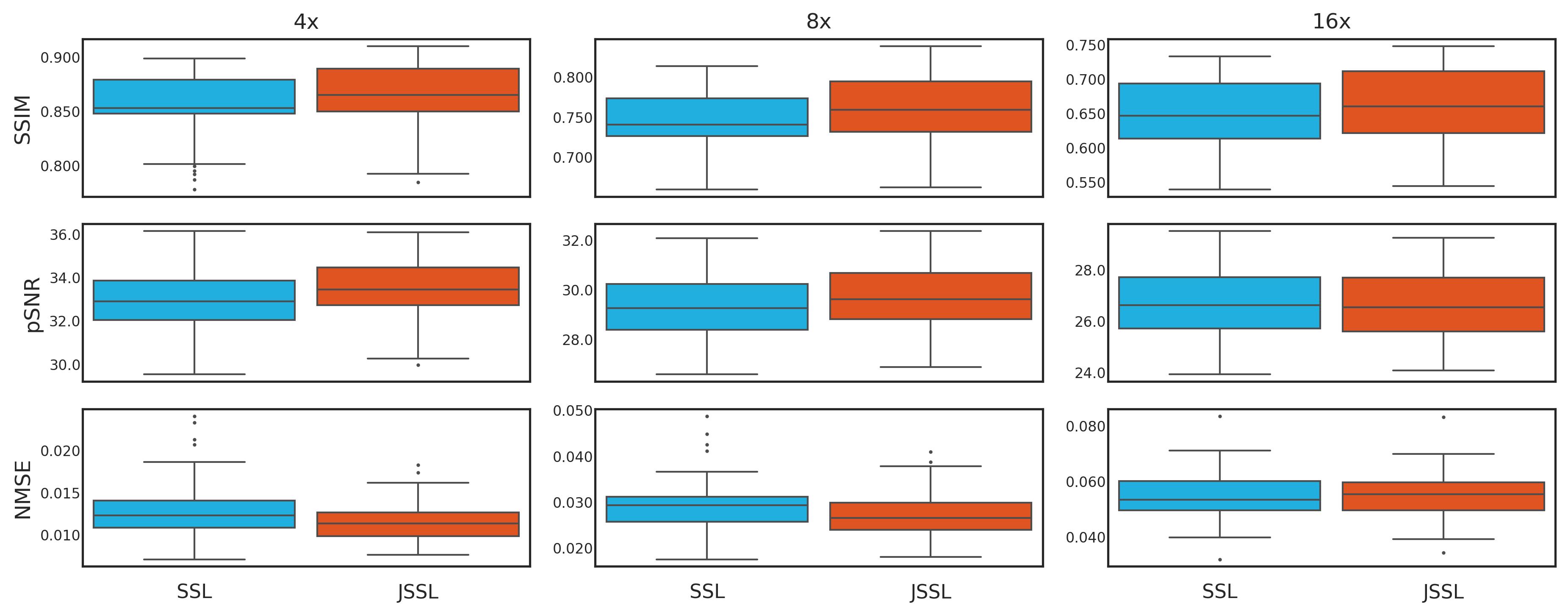}
    \caption{Evaluation metric results for SSL and JSSL methods using an image domain U-Net architecture as a reconstruction network.}
    \label{fig:S2}
\end{figure}

%% file: fig3.tex
\begin{figure}[!ht]
    \centering
    \includegraphics[width=\textwidth]{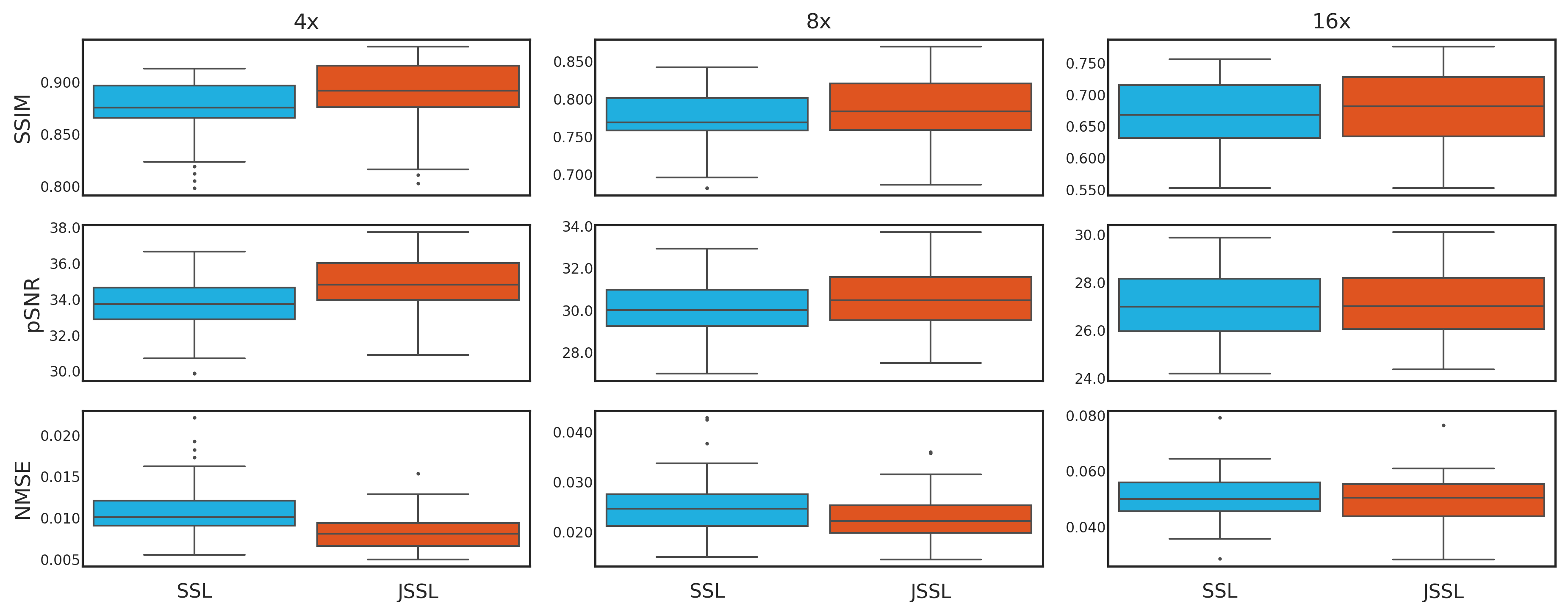}
    \caption{Evaluation metric results for SSL and JSSL methods using an E2EVarNet architecture as a reconstruction network.}
    \label{fig:S3}
\end{figure}

%% file: appendix4.tex
\section{Additional Figures}
\label{sec:ap4}

\subsection{Comparison Studies}
\input{fig4}
\input{fig5}

\break
\clearpage
\subsection{Ablation Studies}


\begin{afterpage}{\clearpage} 
\input{fig6}
\end{afterpage}
\thispagestyle{empty}
\clearpage

\clearpage
\begin{afterpage}{\clearpage} 
\input{fig7}
\end{afterpage}
\thispagestyle{empty}
\clearpage

\clearpage
\begin{afterpage}{\clearpage} 
\input{fig8}
\end{afterpage}
\thispagestyle{empty}
\clearpage

%% file: fig4.tex
\begin{figure}[!ht]
    \centering
    \includegraphics[width=\textwidth]{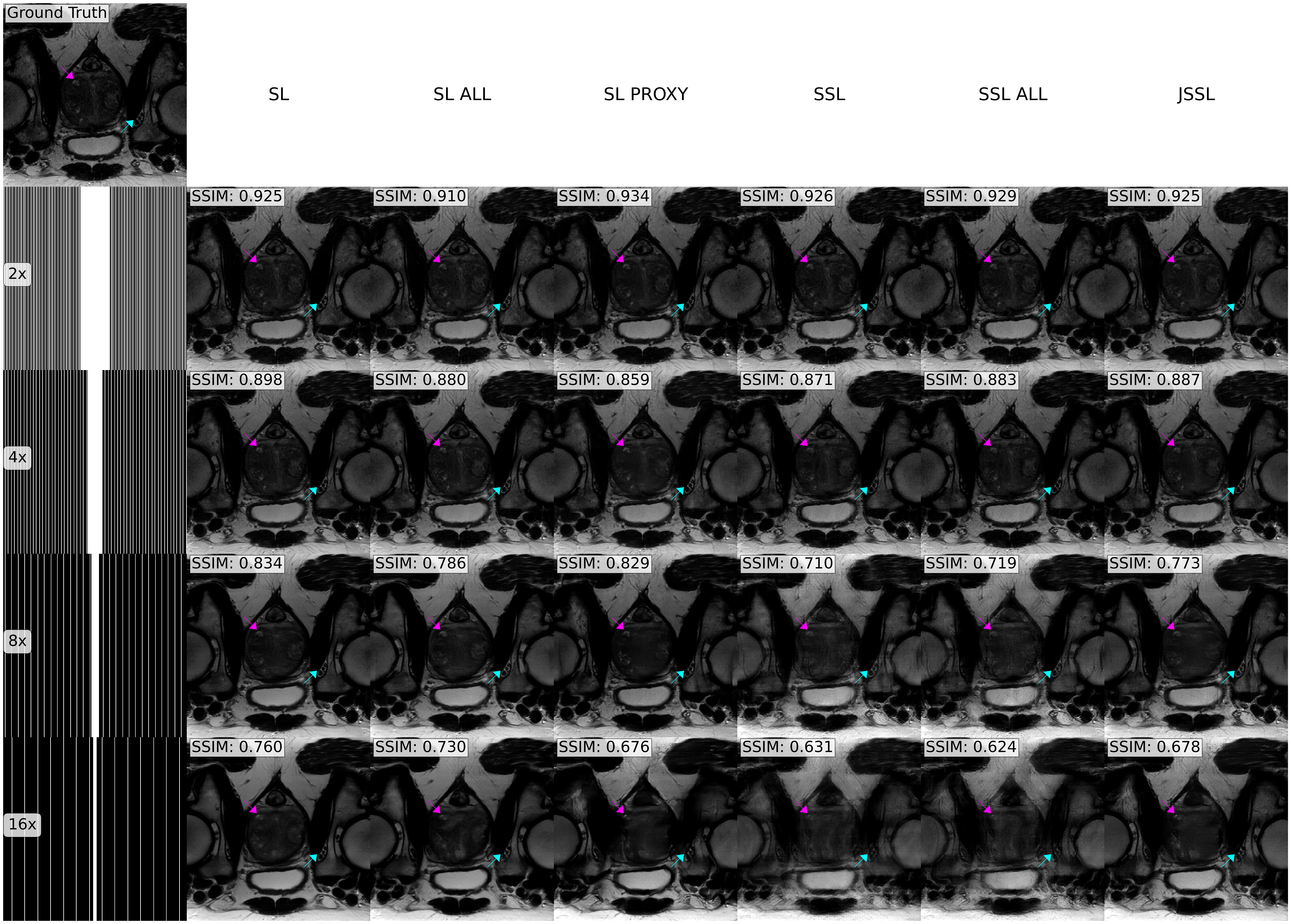}
    \caption{Example reconstructions of a prostate slice subsampled at different acceleration factors (left-most column) from the test set (experiment set \textbf{A}) from each training setup in the comparison studies (Section 4.1 of the main paper) visualized against the ground truth. Arrows point to regions of interest.}
    \label{fig:S4}
\end{figure}

%% file: fig5.tex
\begin{figure}[!ht]
    \centering
    \includegraphics[width=\textwidth]{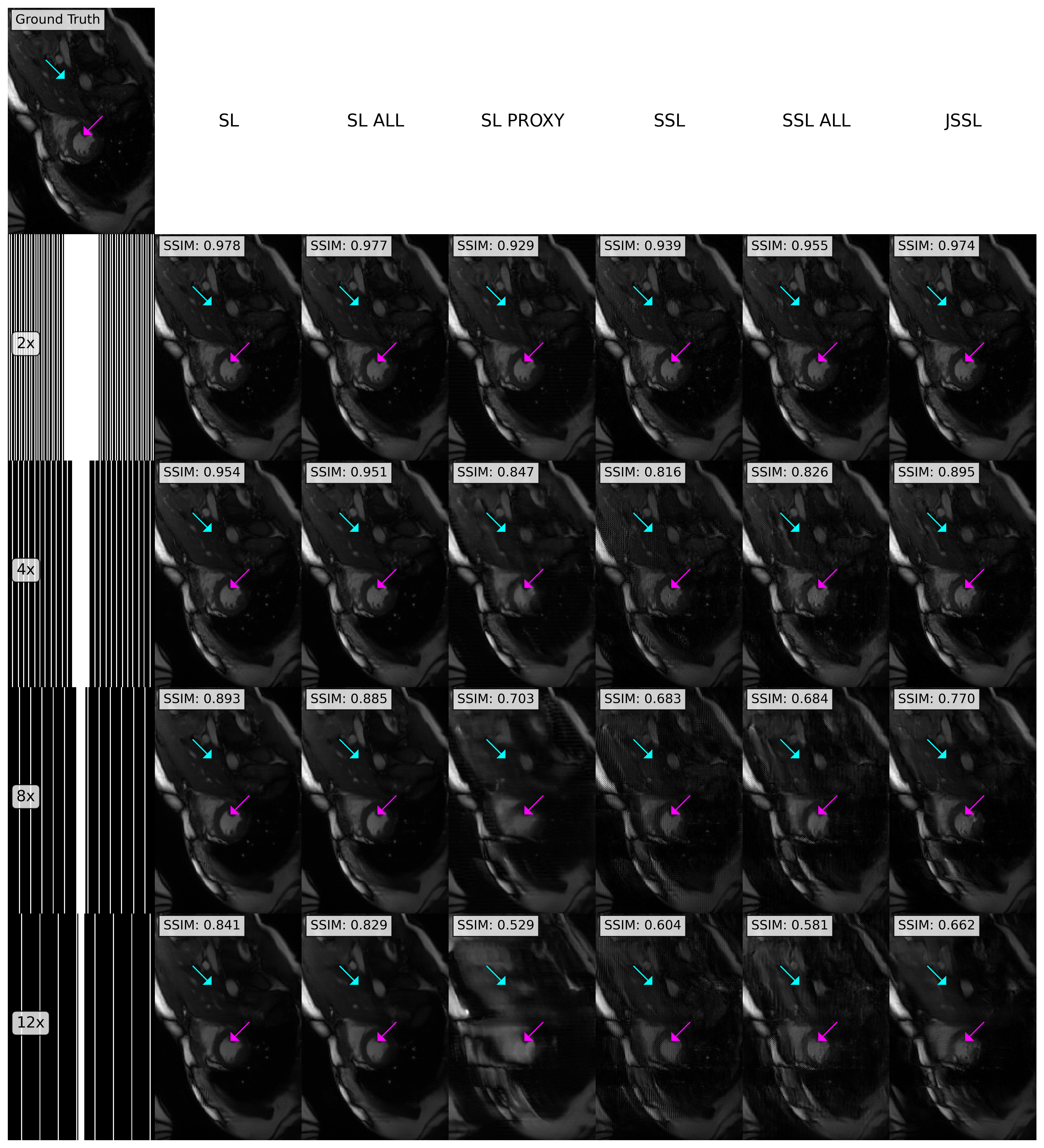}
    \caption{Example reconstructions of a cardiac slice subsampled at different acceleration factors (left-most column) from the test set (experiment set \textbf{B}) from each training setup in the comparison studies (Section 4.1 of the main paper) visualized against the ground truth. Arrows point to regions of interest.}
    \label{fig:S5}
\end{figure}

%% file: fig6.tex
\begin{figure}[!ht]
    \centering
    \includegraphics[width=\textwidth]{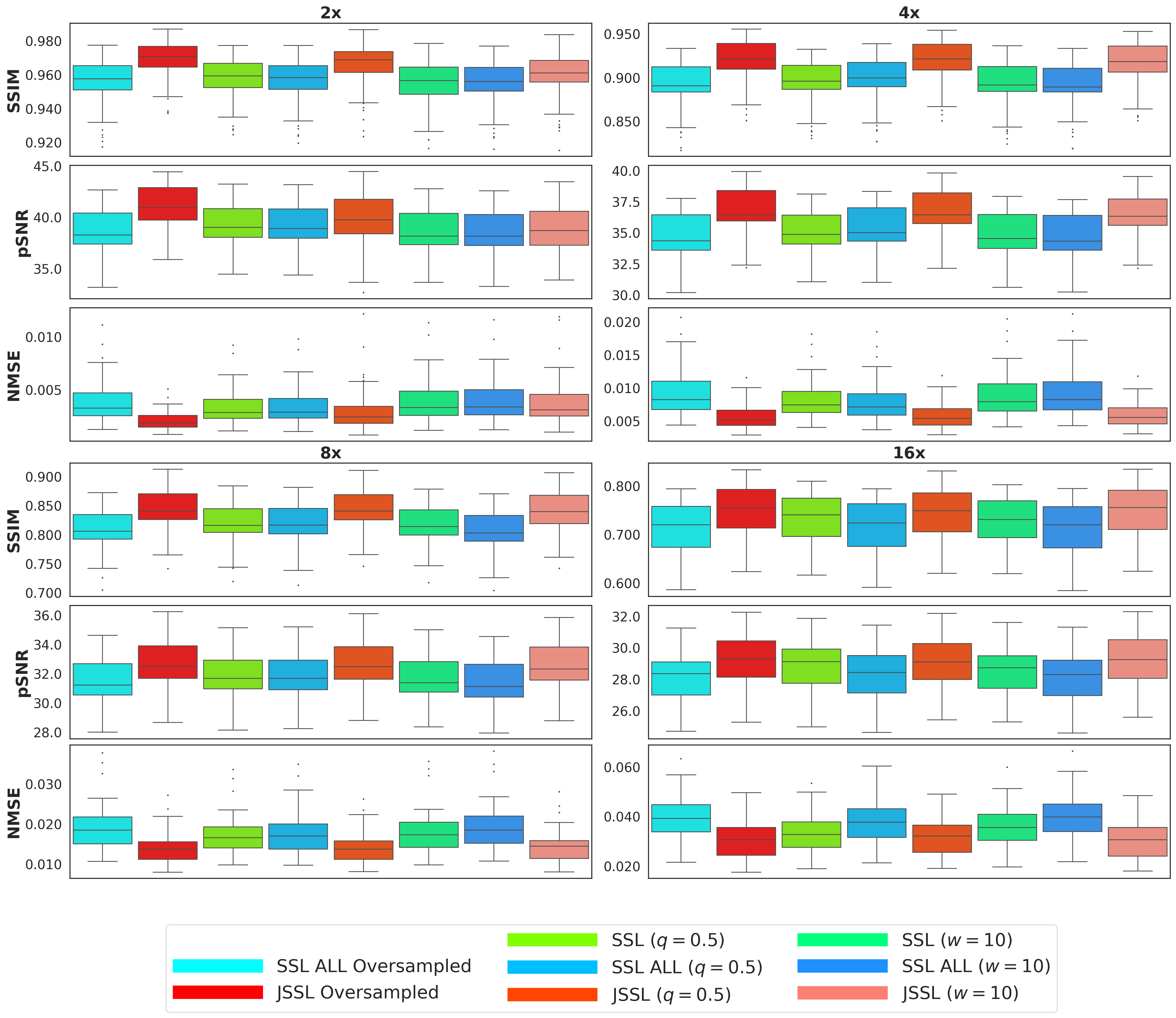}
    \caption{Evaluation results for varying JSSL and SSL setups for the alternative configurations studies of Section 4.2 of the main paper.}
    \label{fig:S6}
\end{figure}

%% file: fig7.tex
\begin{figure}[!ht]
    \centering
    \includegraphics[angle=90,width=0.67\textwidth]{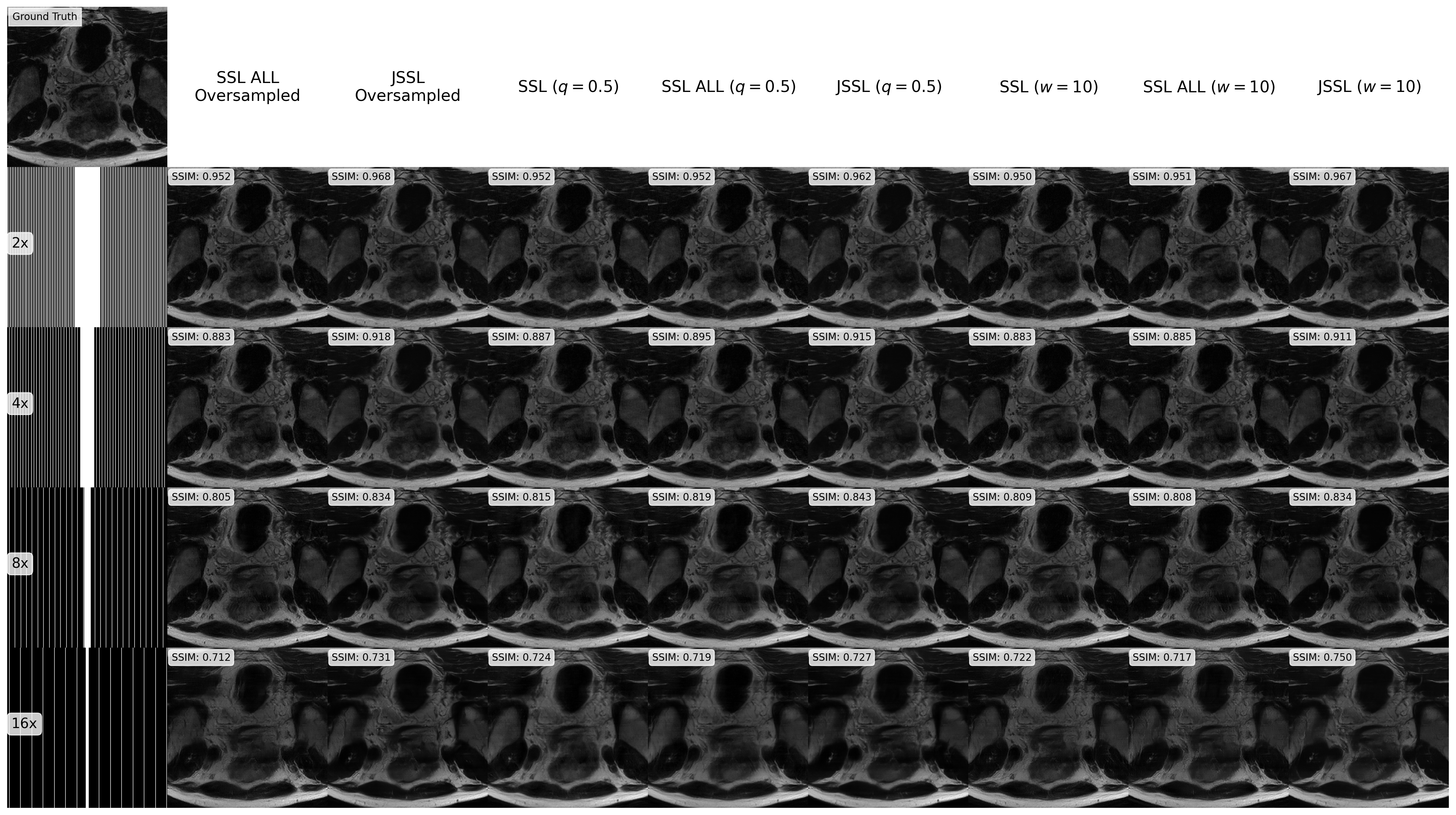}
    \caption{Example 1 of reconstructions of a prostate slice subsampled at different acceleration factors (left-most column) from the test set from each training setup in the alternative configurations studies (Section 4.2 of the main paper) visualized against the ground truth. Arrows point to regions of interest.}
    \label{fig:S7}
\end{figure}

%% file: fig8.tex
\begin{figure}[!ht]
    \centering
    \includegraphics[angle=90,width=0.67\textwidth]{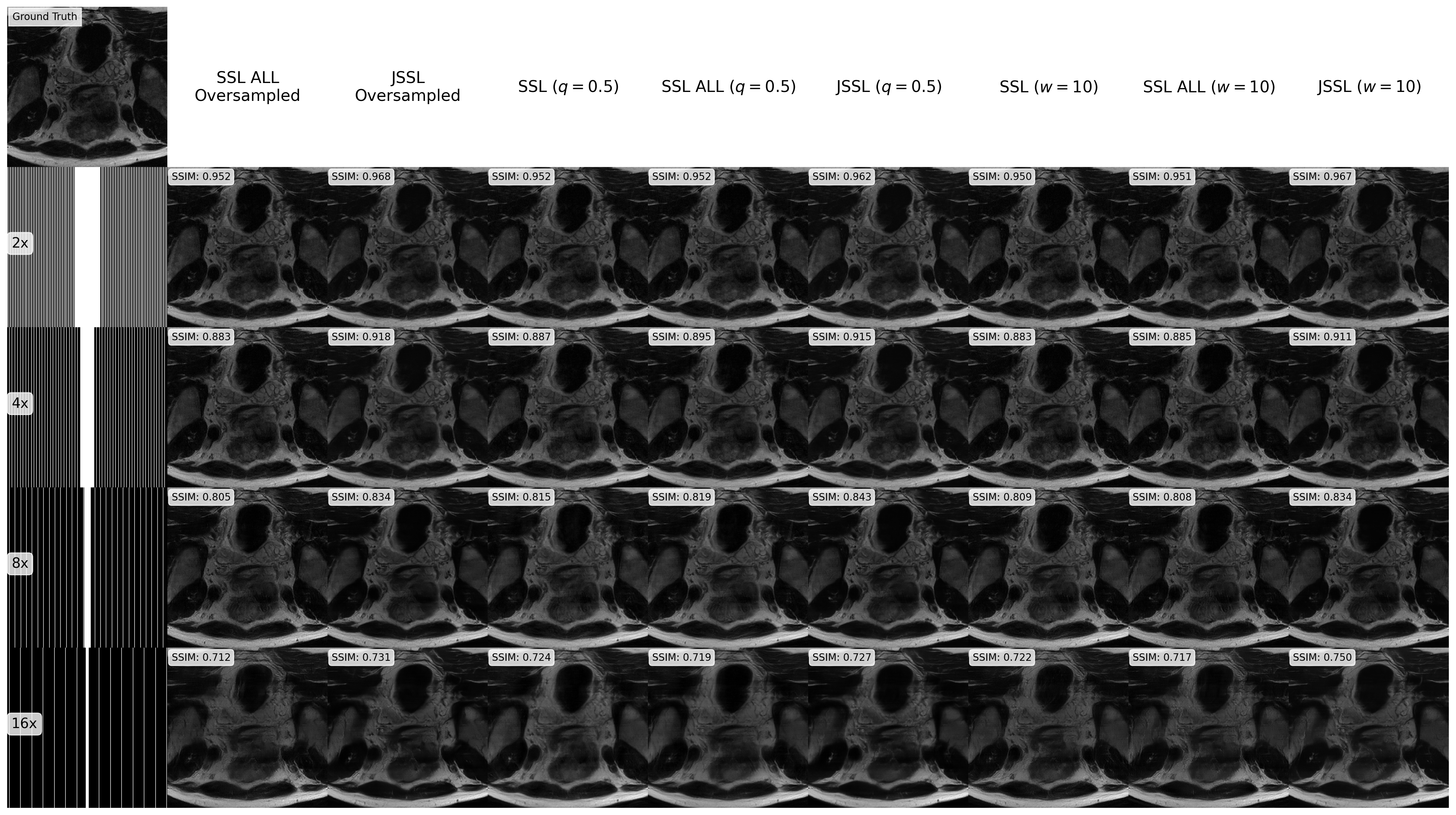}
    \caption{Example 2 of reconstructions of a prostate slice subsampled at different acceleration factors (left-most column) from the test set from each training setup in the alternative configurations studies (Section 4.2 of the main paper) visualized against the ground truth. Arrows point to regions of interest.}
    \label{fig:S8}
\end{figure}